\documentclass[acmsmall,screen]{acmart}

\acmJournal{PACMPL}
\acmVolume{1}
\acmNumber{CONF} 
\acmArticle{1}
\acmYear{2018}
\acmMonth{1}
\acmDOI{} 
\startPage{1}
\setcopyright{none}

\usepackage[utf8]{inputenc}
\usepackage{amsmath}
\usepackage{booktabs}
\usepackage{subcaption}
\usepackage{url}
\usepackage{hyperref}
\usepackage{cleveref}
\usepackage{paralist}
\usepackage{xfrac}
\usepackage{algpseudocode}
\usepackage{lscape}
\usepackage{thm-restate}
\usepackage{code}
\usepackage{sbnf}
\usepackage{paper}

\definecolor{amber}{RGB}{255,191,0}
\definecolor{salmon}{RGB}{250, 128, 114}

\repeatable{proposition}
\repeatable{lemma}
\repeatable{theorem}

\newcommand*\ecoimp{\texttt{eco-imp}\xspace}
\newcommand*\evimp{\texttt{ev-imp}\xspace}
\newcommand*{\raml}{\textsf{RaML}}
\newcommand*{\atlas}{\ensuremath{\mathsf{ATLAS}}}

\begin{document}

\title{Automated Expected Value Analysis of Recursive Programs}

\author{Martin Avanzini}
\orcid{0000-0002-6445-8833}
\affiliation{
  \institution{INRIA Sophia Antipolis Méditerranée}            
  \city{Route des Lucioles - BP 93}
  \country{France}                    
}
\email{martin.avanzini@inria.fr}

\author{Georg Moser}
\orcid{0000-0001-9240-6128}
\affiliation{
  \institution{University of Innsbruck}
  \city{Innsbruck}
  \country{Austria}
}
\email{georg.moser@uibk.ac.at}

\author{Michael Schaper}
\orcid{0000-0002-4795-0615}
\affiliation{
  \institution{Build Informed}
  \city{Innsbruck}
  \country{Austria}
}
\email{mschaper@posteo.net}

\begin{abstract}
In this work, we study the fully automated inference of expected result values of probabilistic programs in the presence of natural programming constructs such as procedures, local variables and recursion. While crucial, capturing these  constructs becomes highly non-trivial. The key contribution is the definition of a term representation, denoted as $\ET{\cdot}$, translating a pre-expectation semantics into first-order constraints, susceptible to automation via standard methods. A crucial step is the use of logical variables, inspired by previous work on Hoare logics for recursive programs. Noteworthy, our methodology is not restricted to tail-recursion, which could unarguably be replaced by iteration and wouldn't need additional insights.
We have implemented this analysis in our prototype~\evimp. We provide ample experimental evidence of
the prototype's algorithmic expressibility.
\end{abstract}

\keywords{static program analysis, probabilistic programming, expected value analysis, weakest pre-expectation semantics, automation}

\maketitle

\renewcommand{\bnfcmt}[1]{\emph{(#1)}}
\section{Introduction}
\label{sec:introduction}

The \emph{verification} of the \emph{quantitative} behaviour of probabilistic programs is a highly active field of research, motivated partly by the recent success of machine learning methodologies (see eg.~\cite{BatzKKMN19,EberlHN20,VasilenkoVB22}).
Without verification, bugs may hide in correctness arguments and subsequent implementations, in particular,
as reasoning about probabilistic programs (or data structures for that matter) is highly non-trivial and error prone.

Instead of verifying quantitative program behaviour semi-automatically, it would be desirable to fully automatically \emph{infer} such estimates.
For example, the goal of inference could be an approximate but still precise computation of \emph{expected (amortised) costs} (eg.~\cite{AMS20,WangKH20,LMZ:2022}) or \emph{quantitative invariants} (eg.~\cite{WangHR18,BaoTPHR:2022}).%
\footnote{The need for automated techniques that analyse eg. the computational cost of code has also been recognised in large software companies. For example at Facebook, one routinely runs a cost analysis on the start-up routines in order to ensure a fast loading of the Facebook web page~\cite{DFLO:2019}.}
The goal of inference is to compute approximations that are as precise as possible, which requires the use of optimisation techniques. As there are now powerful optimising constraint solvers available, such as Z3~\cite{MouraB08} or OptiMathSAT~\cite{SebastianiT20}, we may even hope to automatically conduct optimisations that earlier required intricate analysis with pen and paper, cf.~\cite{LMZ:2021,LMZ:2022}.

The central contribution of our work is the (fully) automated inference of expected result values of probabilistic programs
in the presence of natural programming constructs such as procedures, local variables and recursion.
While crucial, capturing these  constructs becomes highly non-trivial.
To analyse recursive procedures successfully, it is required to properly model the call-stack.
In particular, for each (recursive) procedure call the program
context may change, that is, the analysis of (general) recursion requires
some form of parametricity.%
\footnote{In the context of automated cost analysis this problem has been partly addressed through the notion of \emph{resource parametricity}~\cite{Hoffmann11}.}
To the best of our knowledge, prior work fails in general to deal with these natural constructs.

Upper invariants (on the expected value) constitute a liberalisation of exact quantitative invariants. This approximated rendering of invariants allows for easier (and thus more powerful) automation. Such an analysis may also act as a stepping stone towards the incorporation of support for automation into ITPs.
To wit, a recent and partly motivating work is~\citet{VasilenkoVB22}. We take up one
of their motivating examples---see Listing~\ref{fig:motivation}(\subref{lst:balls}) below---and show how a variant of
this example becomes susceptible to full automation.

Our analysis is based on a weakest pre-expectation semantics~\cite{MM05} (see also~\cite{GKM14,KKMO:ACM:18})---an axiomatic semantic in Dijkstra's spirit---for a simple imperative language \PWHILE\ endowed with sampling instructions,
non-deterministic choice, lexically scoped local variables, and, crucially, recursive procedures.
On a conceptual level, this semantics---denoted as $\et{\prog}$ for a program $\prog$---capture the expected behaviour of probabilistic programs.

Automation of a weakest pre-expectation semantics in the context of recursive procedures is highly non-trivial.%
\footnote{Automation is here understood here as ``push-button'' automation; no user-interaction is required.}
Technically, this is due to the fact that well-definedness of the semantics requires the existence of \emph{higher-order} fixed-points.
Full automation, however, requires the inference of (upper) bounds on closed-forms of such fixed-points.
We overcome this challenge through the definition of a suitable \emph{term representation}---denoted as $\ET{\prog}$---of the aforementioned pre-expectation semantics $\et{\prog}$. This syntactic representation translates the pre-expectation semantics into
first-order constraints, susceptible to automation via standard methods.
To provide for the aforementioned parametricity in this analysis, the use of logical variables---inspired by the work on Hoare logics for recursive programs~\cite{Kleymann99}---is essential.
We make use of a template approach and employ Z3 as suitable optimising SMT solver. Perhaps the closest comparison is to work on amortised cost analysis of functional languages (eg.~\cite{WangKH20,LMZ:2022}).

\paragraph{Contributions.}

In sum, we present a novel methodology for the
\emph{automated} \emph{expected value} analysis of non-deterministic, probabilistic and recursive programs.
Our starting point is a natural \emph{weakest pre-expectation semantics} for
\PWHILE\ in the form of an \emph{expectation transformer} $\et{\prog}$, capturing also
natural program constructs such as lexically scoped
local variables, procedure parameters, unrestricted return statements, etc.
Our main contribution lies in a novel \emph{term representation} of
the aforementioned expectation transformer, denoted as $\ET{\prog}$.
Its definition follows the pattern of $\et{\prog}$, but notably differs
in the definition of procedure calls and loops,
where we replace the underlying fixed-point constructions via suitably constrained first-order templates.
Through this step, we manage to tightly over-approximate the precise, but higher-order semantics via
first-order constraint generation susceptible to automation.
Second, we establish a new and original \emph{automation} of this quantitative analysis
in our prototype implementation~\evimp.

\paragraph{Outline.}
In Section~\ref{Overview} we provide a bird's eye view
on the contributions of this work. In Section~\ref{PWhile} we detail the syntactic
structure of our language \PWHILE\ and provide the definition of the aforementioned
weakest pre-expectation semantics. Section~\ref{Inference} constitutes the main technical
part of the work, detailing the conception and definition of a term representation
of the weakest pre-expectation semantics, susceptible to automation. Conclusively in Section~\ref{Implementation}
we discuss implementation choices for our prototype implementation, the chosen benchmark suites and the experimental evaluation. Finally, we conclude with related works and future work in Sections~\ref{Related} and \ref{Conclusion}, respectively.
Omitted proofs can be found in the Supplementary Material.


\section{Expected Value Analysis, Automated}
\label{Overview}

\begin{figure*}
\caption{Textbook Examples on Expected Value Analysis}
\label{fig:motivation}  
\smallskip  
\centering
\begin{subfigure}[b]{0.34\textwidth}
\begin{lstlisting}[style=pwhile,style=framed,mathescape]
def balls(n):
  var $b \ass 0$
  if (n > 0) {
    b $\ass$ balls(n-1);
    if (Bernoulli($\sfrac{1}{5}$)) {b $\ass$ b + 1}
  };
  return b
\end{lstlisting}
\caption{Balls in a single bin.}
\label{lst:balls}
\end{subfigure}
\hspace{1ex}\hfill
\begin{subfigure}[b]{0.25\textwidth}
\begin{lstlisting}[style=pwhile,style=framed,mathescape]
def throws():
  if (Bernoulli($\sfrac{1}{5}$)) {
    return 1
  } else {
    return (1 + throws())
  }

\end{lstlisting}
\caption{Throws for a hit.}
\label{lst:number}
\end{subfigure}
\hspace{1ex}\hfill
\begin{subfigure}[b]{0.35\textwidth}
\begin{lstlisting}[style=pwhile,style=framed]
def every(i):
  if (0 < i <= 5) {
    if (Bernoulli($\sfrac{\tt i}{5}$)) {i $\ass$ i - 1};
    return (1 + every(i))
  } else {
    return 0
  }
\end{lstlisting}
\caption{Every bin contains at least one ball.}
\label{lst:every}
\end{subfigure}
\vspace{-3mm}
\end{figure*}

We motivate the central contribution of our work: an \emph{automated} expected result value
analysis of non-deterministic, probabilistic programs, featuring \emph{recursion}.
First, we present three textbook examples (and their encoding in \PWHILE) motivating the interest and importance
of an \emph{expected value} analysis, in contrast to eg.\ an expected cost analysis
(see eg.~\cite{KKMO:ACM:18,NgoCH18,WFGCQS:PLDI:19,AMS20,WangKH20,LMZ:2022})
or the inference of quantitative invariants (see eg.~\cite{KatoenMMM10,WangHR18,BaoTPHR:2022}).
Second, we emphasise the need for formal verification techniques going beyond pen-and-paper
proofs and provide the intuition behind the thus chosen methodology of so-called \emph{expectation transformers}.
Third, we detail challenges posed by \emph{recursive procedures} for the formal definition of
our expectation transformers.
Finally, we highlight the challenges of automated verification techniques in this context. 

\paragraph*{Textbook examples.}

To motivate expected value analysis in general, we study corresponding textbook examples---taken from
Cormen et al.~\cite{Cormen:2009}---and suitable code representations thereof, depicted in Figure~\ref{fig:motivation}.
Consider an unspecified number of bins and suppose we throw balls towards the bins.
Let us attempt to answer the following three questions:
\begin{inparaenum}[(i)]
\item \emph{How many balls will fall in a given bin?}
\item \emph{How many balls must we toss, on average, until a given bin contains a ball?} and finally
\item \emph{How many balls must we toss until every bin contains at least one ball?}
\end{inparaenum}
  
The first two questions can be easily answered given some mild background in probability theory.
If we throw say~$n$ balls, the number of balls per bin follows a binomial distribution.
Assuming we will always hit at least one bin, then the success probability is $\sfrac{1}{bins}$,
where $bins$ denotes the number of bins.
Thus, the expected value of the number of balls in a single bin is given by $\sfrac{1}{bins} \cdot n$.
Success in the second problem follows a geometric distribution. Thus the answer is $\frac{1}{\sfrac{1}{bins}} = bins$. But the last one is more tricky and involves a more intricate argumentation. Note that this problem
and thus the procedure \pwhile!every!
is equivalent to the \emph{Coupon Collector} problem, cf.~\cite{MitzenMacherU05,Cormen:2009,AMS20}.
Following the argument in~\cite[Chapter~5.4]{Cormen:2009}, we say that we have a ``hit'', if we have successfully hit a specific bin. Using the notion of hits, we can split the task into stages, each corresponding to a completed hit. If all stages are complete, we are done. To complete stage $k$, we need to hit the $k^{th}$ bin. Thus, the success
probability to complete this stage is given as $\sfrac{bins-k+1}{bins}$.
Now let $n_k$ denote the number of throws in the $k^{th}$ stage. Thus, the total number required to fill the bins is $n = \sum_{k=1}^{bins} n_k$. As $\EEop[n_k] = \sfrac{bins}{bins-k+1}$, we obtain from linearity
of expectations
\begin{equation*}
  \EEop[n] = \EEop[ \sum_{k=1}^{bins} n_k] = \sum_{k=1}^{bins} \EEop[n_k]
  = \sum_{k=1}^{bins} \frac{bins}{bins-k+1} = bins \cdot \sum_{k=1}^{bins} \frac{1}{k}
  \tpkt
  \tkom
\end{equation*}
Further, we obtain $bins \cdot \sum_{k=1}^{bins} \sfrac{1}{k} \in \Theta(bins \cdot \log(bins))$,
utilising that the \emph{harmonic number} $H_{bins} = \sum_{k=1}^{bins} \sfrac{1}{k}$
is asymptotically bounded by $\log(bins)$, cf.~\cite{GKP1994}.

\paragraph*{Probabilistic, recursive procedures.}

As depicted in Listings~\ref{fig:motivation}(\subref{lst:balls}--\subref{lst:every}), it is easy to encode the above questions as (probabilistic, recursive) procedures.%
\footnote{We emphasise that procedure \pwhile!balls! constitutes a variant of the leading example in~\cite{VasilenkoVB22}. Using refinement types, \citet{VasilenkoVB22} prove that the expected value of $b$ is $p \cdot n$, where $p$ denotes
the probability of hitting the given bin ($p=\sfrac{1}{5}$ in our rendering). Their analysis, however, is only semi-automated.}
The encoding fixes the number of bins to five.
The command \pwhile!Bernoulli($p$)! draws a value from a Bernoulli distribution, returning $1$ with probability $p$ and $0$ with probability $(1-p)$.%
\footnote{For ease of encoding,  Listing~\ref{fig:motivation}(\subref{lst:every}) employs the sampling from \pwhile!Bernoulli($\sfrac{i}{5}$)!. Setting $i = bins-k+1$ and $bins=5$, we recover the textbook argumentation.}
Our simple imperative language \PWHILE\ follows the spirit of Dijkstra's \emph{Guarded Command Language}, see Section~\ref{PWhile} for the formal details. Apart from randomisation and recursion, the encoding makes---we believe natural---use of local variables, formal parameters and return statements.

The textbook questions above now become (possibly intricate) questions on program behaviours, that is,
on the \emph{expected value} of program variables wrt.\ the memory distributions in the final program
state. Expected value analysis encompasses expected cost analysis, as we can often, that is, for almost surely terminating programs, represent program costs through a dedicated counter. Similarly, upper bounds to expected values
as considered here, form approximations of quantitative invariants. In the experimental
validation of our prototype implementation, we will see that wrt.\ the latter often our bounds are in fact exact, that is,
constitute \emph{invariants} (see Section~\ref{Implementation}).

It is well-known that in general pen-and-paper analyses of the expected value of programs are not
an easy matter---as we have for example seen in the answer to the last question above---%
and more principled methodologies are essential, in particular if our focus is on (full) automation.
To this avail, we build upon earlier work
on \emph{weakest pre-expectation semantics} of (probabilistic) programs and develop an
\emph{expectation transformer} for \PWHILE.

\paragraph{Expectation Transformers.}

Predicate transformer semantics---introduced in the seminal works of
\citet{Dijkstra75}---map each program statement to a
function between two predicates on the state-space of the program.
Their semantics can be viewed as a reformulation of Floyd–Hoare
logic~\cite{Hoare69}.
Subsequently, this methodology has been extended to randomised programs by replacing predicates
with so-called expectations---real valued functions on the program's state-space $\Mem{V}$---%
leading to the notion of \emph{expectation transformers} (see~\cite{MM05,KKMO:ACM:18})
and the development of \emph{weakest pre-expectation semantics}~\cite{GKM14}.

In the following, we impose a pre-expectation semantics on programs in~\PWHILE, thereby providing
a formal definition of the expected value of program values. Adequacy of such semantics is well understood in the literature, for example wrt. operational semantics based on Markov Decision Processes~\cite{KKMO:ACM:18} or
probabilistic abstract rewrite systems~\cite{AMS20}.
The difference to our work seems incremental to us. Thus, our starting point are pre-expectation transformer
semantics, permitting us to focus on inference.

For a given command $\cmd$, its \emph{expectation transformer} $\et{\cmd}$
(detailed in Figure~\ref{fig:et} in Section~\ref{PWhile})
maps a post-expectation $f$ to a pre-expectation, measuring the expected value that $f$ takes
on the distribution of states resulting from running~$\cmd$, as a function in the initial state of $\cmd$.
Slightly simplifying, the expectation transformer assigns semantics to \emph{commands}
\[
  \et{\cmd} : (\Mem{V} \to \Rext) \to (\Mem{V} \to \Rext) \tpkt
\]
In the special case where~$f$ is a predicate, ie.\ a $\{0,1\}$ valued function on memories,
$\et{\cmd}\app f \app \mem$ yields the probability that $f$ holds after completion of~$\cmd$ for any initial state~$\mem$;
thereby generalising Djikstra's and Nielson's weakest-precondition transformer, cf.~\cite{Nielson87,KKMO:ACM:18}.
This definitions is probably most easily understood as a denotational semantics in continuation passing style, with post-expectation $f$ interpreted as \emph{continuation}, cf.~\cite{FW:2008,ABD:2021}.

In this reading, $f \colon \Mem{V} \to \Rext$ measures a quantity on the continuations outcome
given its inputs---in expectation---
while $\et{\cmd} \app f$ performs a backward reasoning lifting $f$ to a function
on initial states of $\cmd$. (In the following, we use ``continuation'' and ``expectation'' interchangingly.)

Let us consider the procedure \pwhile!balls! depicted in Figure~\ref{fig:motivation}(\subref{lst:balls}).
As mentioned, the final value of $b$, the number of balls hitting a dedicated bin,
follows a binomial distribution. In each recursive call the success probability is $\sfrac{1}{5}$. Thus, the expected value of the number of balls in a single bin is $\sfrac{1}{5} \cdot n$.
To provide a bird's view on the working of expectation transformers in this context, we verify this simple equality
in the sequel. 
Calculating the pre-expectation from the continuation $f$ is straightforward, as long
as the program under consideration contains neither recursive calls nor loops.
To wit, let $\normf{b}$ be the real-valued function that measures (the positive value of) variable $\pwhile!b!$ in a given memory.
Then
\begin{align*}
  \et{\pwhile!if (Bernoulli($\sfrac{1}{5}$)) \{b=b+1\}!}\app \normf{b}
  = \sfrac{1}{5} \cdot \et{\pwhile!b=b+1!} \app \normf{b} \mathrel{\plusf} \sfrac{4}{5} \cdot \normf{b}
  = \sfrac{1}{5} \cdot \normf{b+1} + \sfrac{4}{5} \cdot \normf{b}
  \tpkt
\end{align*}
Here, addition and scaling by a constant on expectations should should be understood point-wise.
Notice how this expression captures the fact \pwhile!b! is incremented only with probability $\sfrac{1}{5}$.
\emph{Composability} of the transformer extends this kind of reasoning to sequences of command,
in the sense that command composition is interpreted as the composition of the corresponding expectation transformers. For straight-line programs $\cmd$, $\et{\cmd} \app f$ can be computed
in a bottom up fashion.

\paragraph{Loops and Recursive Procedures.}

When dealing with loops or recursive procedures though, the definition of $\et{\cmd}$ becomes
more involved. As usual in giving denotational program semantics,
the definition of the expectation transformer of these self-referential program constructs is based on a
fixed-point construction, cf.~\cite{MM05,KatoenMMM10,KKMO:ACM:18,WangHR18,AMS20,BaoTPHR:2022}.
This permits attributing precise semantics to programs.
In our setting, \emph{procedures} $\fn$ are interpreted via an expectation transformer
\[
  \et{\fn} : (\Val \times \GMem \to \Rext) \to (\Val^{\ar{\fn}} \times \GMem \to \Rext)
  \tkom
\]
where $\GMem$ the set of global memories.
Following the definition of $\fn$, $\et{\fn}$ turns a continuation, parameterised in the return value and global memory,
into a pre-expectation in the formal parameters of $\fn$ and the global memory before execution of $\fn$.
The definition (formalised in Section~\ref{PWhile}) is slightly technical to permit recursive definitions and
to ensure proper passing of arguments and lexical scoping.
For a procedure declared as $\cdef{\fn}{\vec{\vx}}{\cmd}$ and a post-expectation $f$,
$\et{\fn} \app f$ is given by the transformer $\et{\cmd}$ associated to its body $\cmd$,
initialising the parameters~$\vec{\vx}$ accordingly
and with the post-expectation $f$ applied whenever $\cmd$ leaves its scope through a return-statement.

Let us re-consider the procedure \pwhile!balls! from Listing~\ref{fig:motivation}(\subref{lst:balls}).
We use \emph{Iverson brackets} $\brac{\cdot}$ to lift predicates on memories to expectations,
that is, $\bracf{\bexpr}(\mem) \defsym 1$, if $\bexpr$ holds in memory $\mem$,
and $\bracf{\bexpr}(\mem) \defsym 0$, otherwise.
Since the program does not make use of global variables,
we elide representation of the empty global memory for brievity.
Thus, for an arbitrary continuation~$f : \Val \to \Rext$,
the expectation transformer of $\pwhile!balls!$ is given by the least functional 
satisfying:\footnote{Note that the functional is also ordered point-wise. Well-definedness of the employed fixed-point construction rests on the observations that expectation transformers form \emph{$\omega$-CPOs} \cite{Winskel:93}.}
\begin{equation}
  \label{eq:et-balls}
  \tag{\dag}
  \et{\pwhile!balls!} \app f = \lam{n}{[n>0] \cdot \et{\pwhile!balls!} \app \bigl(\lam{b}{\sfrac{1}{5}\cdot f(b+1) + \sfrac{4}{5}\cdot f(b)}\bigr) \app (n-1) + [n\leq 0]\cdot f(0) }
  \tkom
\end{equation}
where the continuation $f$ passed to the recursive call of $\et{\pwhile!balls!}$
is computed as above. Thus, the continuation to the call to \pwhile!balls! probabilistically updates the returned value.
Arguing inductively, we see that $\et{\pwhile!balls!} \app f_{r} = \lam{n}{\sfrac{1}{5} \cdot n + r}$
for any continuation of the form $f_{r} \defsym \lam{b}{b+r}$, where $r$ denotes a non-negative real.
Since $f_{0}$ measures the return value, we conclude that
on argument $n$, $\et{\pwhile!balls!}$ returns a value of $\sfrac{1}{5} \cdot n$ in expectation.
Conclusively, we re-obtain that a fifth of the thrown balls will fall in
the considered bin.

In the same vein, the analysis of the return value of \pwhile!throws! and \pwhile!every!---in expectation---is performed for Listings~\ref{fig:motivation}(\subref{lst:number}) and~\ref{fig:motivation}(\subref{lst:every}), respectively.
In this fashion, the translation of the above pen-and-paper analysis for the textbook examples
into the formalism of an expectation transformer results into a rigorous and principled methodology.
This builds a sound (and suitable) foundation for subsequent automation.

\paragraph{Automation.}
The above calculation of the expected return value of procedure \pwhile!balls! gives evidence
on the advantages of a formal methodology over an ad-hoc pen-and-paper analysis.
The crux of turning such a calculus into a fully automated analysis lies in automatically
deriving closed forms for expected values of loops and recursive procedures.
Related problems have been extensively studied in the literature, eg.~\cite{PR04,contejean:2005,FGMSTZ:SAT:07,SZV:2014}.
One prominent approach lies in assigning \emph{templates} to expected value functions, by which
the definition of the expectation transformer can be reduced to a set of constraints treatable with off-the-shelf SMT solvers, like Z3~\cite{MouraB08}.

Following this approach, we express values as
linear combination $\sum_i c_i \cdot b_i$ of \emph{base functions} (or \emph{norms}) $b_i$
with variable coefficients $c_i$, mapping program valuations to (non-negative) real numbers. 
Norms serve as a numerical abstraction of memories and encompass a variety
of common abstractions, for example the absolute value of a variable, the difference between two variables
and more generally arbitrary polynomial combinations thereof.
Often, expected values can be computed symbolically on such \emph{value expressions}.
Concerning recursive procedures or loops, the definition of the expectation transform is in essence recursive itself. Rather than computing this fixed-point directly, we make use of Park's theorem~\cite{Wechler92,KKMO:ACM:18} and seek an upper bound in closed-form. In the context of recursive procedure, we also have to model the call-stack appropriately. 

Concerning recursive procedures,
our solution is to represent every higher-order functional $\et{\fn}$ through a pair of first-order templates
$\soa{\termh[]{\fn}}{\termk[]{\fn}}$,
representing pre- and post-expectations respectively.
Concretely, such a template signifies $\et{\fn} \app \termk{\fn} \leq \termh{\fn}$.
Inspired by Hoare-calculi for recursive programs (cf.~\cite{Kleymann99}),
we index templates by \emph{logical variables} $\lv$, kept implicit in the sequel, but emphasised here for clarity of exposition. Thus indeed, a template represents
a \emph{families of} pre-/post-expectations.
To wit, let us revisit procedure \pwhile!balls! in Listing~\ref{fig:motivation}(\subref{lst:balls}) once more.
We take the templates:
\begin{align*}
\termh[\lv]{\pwhile!balls!} & \defsym \lam{n}{c_{0} + c_{1} \cdot \normf{n} + c_{2} \cdot \lv}
  & \termk[\lv]{\pwhile!balls!} & \defsym \lam{b}{\normf{b} + \lv}
                             \tkom
\end{align*}
Notice that the logical variable $\lv$ should only be instantiated non-negatively, so as to ensure
that these templates indeed capture (non-negative) expectations.
Thus, the indended meaning of this template for \pwhile!balls!
is captured by the constraint
\begin{align}
  \label{eq:mainconstr}
  \tag{\ddag}
  \forall \lv \geq 0.\, \forall n.\, \et{\pwhile!balls!} \app \termk[\lv]{\pwhile!balls!} \app n
  \leq
  \termh[\lv]{\pwhile!balls!} \app n
  \tkom
\end{align}
Symbolic evaluation of the left-hand side---making use of the template for calls to \pwhile!balls!---%
can now be used to sufficiently constrain the undetermined coefficients.
Concretely~\eqref{eq:mainconstr} is representable via the following 
three constraints, where the variables $b,n$ and $\lv$ are universally quantified,
while the unknowns $c_{0}$, $c_{1}$, $d_{0}$ and $d_{1}$ are existentially quantified.
\begin{align}
  & 0 \leq \lv \entails \sfrac{1}{5}\cdot \normf{b+1} + \sfrac{4}{5}\cdot \normf{b} + \lv \leqf \normf{b} + (d_{0} + d_{1} \cdot \lv) \label{balls-c1}\tag{c1} \\
  & 0 \leq \lv \entails 0 \leqf d_{0} + d_{1} \cdot \lv \label{balls-c2}\tag{c2} \\
  & 0 \leq \lv \entails \bracf{n > 0} \cdot (c_{0} + c_{1} \cdot \normf{n-1} + c_{2}\cdot(d_{0} + d_{1} \cdot \lv)) + \bracf{n \leq 0} \cdot \lv \leqf c_{0} + c_{1}\normf{n} + c_{2}\cdot \lv \label{balls-c3}\tag{c3}
    \tpkt
\end{align}
To see how these constraints have been formed, recall the fixed-point equation for $\et{\pwhile!balls!}$
in \eqref{eq:et-balls}, instantiating $f = \termk[\lv]{\pwhile!balls!}$.
Constraint~\eqref{balls-c1} expresses that the post-expectation
\[
\lam{b}{\sfrac{1}{5}\cdot \termk[\lv]{\pwhile!balls!}(b+1) + \sfrac{4}{5}\cdot \termk[\lv]{\pwhile!balls!}(b)}
= \lam{b}{\sfrac{1}{5}\cdot \normf{b+1} + \sfrac{4}{5}\cdot \normf{b} + \lv}
\tkom
\]
passed to the call of $\et{\pwhile!balls!}$ in~\eqref{eq:et-balls} is over-approximated by instance $\termk[{d_{1} + d_{1} \cdot \lv}]{\pwhile!balls!}$,
taking a (to be further determined) instantiation $\lv \mapsto d_{1} + d_{1} \cdot \lv$ of the logical variable $\lv$.
This substitution is guaranteed non-negative (and thus well-defined) by constraint~\eqref{balls-c2}.
Thus taking~\eqref{balls-c1} into account, y \eqref{eq:mainconstr} e have
\[
  \et{\pwhile!balls!} \app (\lam{b}{\sfrac{1}{5}\cdot \termk[\lv]{\pwhile!balls!}(b+1) + \sfrac{4}{5}\cdot \termk[\lv]{\pwhile!balls!}(b)}) \app (n - 1)
  \leq \termh[{d_{0} + d_{1} \cdot \lv}]{\pwhile!balls!} \app (n-1)
  \tpkt
\]
Making use of this over-approximation in \eqref{eq:et-balls},
it is now evident that the final constraint~\eqref{balls-c3} witnesses satisfiability of the main constraint~\eqref{eq:mainconstr}.
These three constraints can be met by taking $c_{1},d_{0} = \sfrac{1}{5}$, $c_{2},d_{1} = 1$, and $c_{0} = 0$.

We re-obtain (unsurprisingly) that the expected return value of procedure \pwhile!balls! is given as $\sfrac{1}{5} \cdot n$.
We emphasise that parameterising templates through logical variables is crucial, even for relatively simple examples
such as \pwhile!balls!. Since the program is not tail-recursive, the
continuation changes after each recursive call.
Parametricity allows us to vary templates across recursive calls, as we have done above through instantiating $\lv$ by ${d_{0} + d_{1} \cdot \lv}$.

We have successfully automated this approach in our prototype implementation~\evimp,
see Section~\ref{Implementation}. Automation is based on our first contribution,
the development of a \emph{term representation} of the expectation transformer inducing an inference method that describes the generation of constraints.
In the context of procedure \pwhile!balls!, this representation reduces the task of finding a functional
satisfying~\eqref{eq:et-balls} to the definition of a set of constraints like~\eqref{eq:mainconstr}, whose solution
yields over-approximations of the function graph of $\et{\pwhile!balls!}$. (See Section~\ref{Inference} for the details.)
Based on this formal development, our tool~\evimp\ proceeds with an analysis as
outlined, and integrates a dedicated constraint solver for solving constraints of the form~\eqref{balls-c1}--\eqref{balls-c3}.
Apart from handling recursive procedures, we have incorporated the constraints for handling loop programs.
Here, we take inspiration from~\cite{AMS20} to guarantee a \emph{modular} (and thus \emph{scalable}) inference of upper bounding functions.

Wrt.\ the three motivating examples in Figure~\ref{fig:motivation}, our tool derives (upper-bounds to) the
corresponding expected return values in milliseconds. The bounds for the procedures in Listings~\ref{fig:motivation}(\subref{lst:balls}) and~(\subref{lst:number}) are precise, respectively. For the procedure \pwhile!every! we, however,
only manage to derive a (sound) upper bound. The latter is not surprising. As shown in the textbook proof the
expected number of throws is given as $\Theta(bins \cdot \log (bins))$ in general. Our template approach does not (yet)
support logarithmic functions. Hence, we cannot hope to derive the precise bound fully automatically.
To the best of our knowledge, our prototype~\evimp is the only existing tool able to fully automatically
analyse the expected result value (or expected cost for that matter) of probabilistic, recursive programs,
if formulated in an imperative language.
The complete evaluation of our prototype implementation is given in Section~\ref{Implementation}.

\section{An Imperative Probabilistic Language}
\label{PWhile}

We consider a small \emph{imperative language} in the spirit of Dijkstra's \emph{Guarded Command Language}.
In particular, the language features
\begin{inparaenum}[(i)]
\item dynamic sampling instructions; 
\item non-deterministic choice, formalised via a  non-deterministic choice operator $\ndci$;
\item nested loops; and
\item crucially recursive procedure declarations.
\end{inparaenum}
In this section, we first formalise the syntax and then endow it with axiomatic, pre-expectation
semantics.

\paragraph{Syntax.}
Let $\Var = \{\vx,\vy,\dots\}$ be a set of (integer-valued) program \emph{variable}.
We fix three syntactic categories of
\emph{Boolean} valued \emph{expressions} $\BExpr \app V$,
\emph{(integer valued) expressions} $\Expr \app V$, and
\emph{sampling instructions} $\SExpr \app V$ over (finitely many) variables $V \subseteq \Var$, respectively.
In the following, $\bexpr$ will range over Boolean, $\expr,\exprtwo$ over integer valued expressions and $\dexpr$ over sampling instructions.
Furthermore, let $\Proc = \{\fn,\fntwo, \dots\}$ be a set
of \emph{procedure (symbols)}.
The set of commands $\Cmd \app V$ over program variables $\vx \in V$ is given by 
\begin{bnf}
  \cmd,\cmdtwo &
  \cskip
  \bnfmid \vx \sample \dexpr
  \bnfmid \vx \sample \fn(\expr_1,\dots,\expr_{\ar{\fn}})
  \bnfmid \cret{\expr}
  \bnfmid \clet{\vx}{\expr}{\cmd}
  & \\
  & \cmd \sep \cmdtwo
  \bnfmid \cif{\bexpr}{\cmd}{\cmd}
  \bnfmid \cwhile{\bexpr}{\cmd}
  \bnfmid \cnd{\cmd}{\cmdtwo}
\end{bnf}

The interpretation of these commands is fairly standard.
The command $\cskip$ acts as a no-op.
In an assignment $\vx \sample \rexpr$, the right-hand side $\rexpr$ evaluates to a distribution,
from which a value is sampled and assigned to $\vx$.
As right-hand side $\rexpr$, we permit either built-in sampling expressions $\dexpr \in \SExpr$ such as
$\ccall{unif}{lo,hi}$ for sampling an integer uniformly between constants $lo$ and $hi$,
or calls to user-defined procedures~$\fn \in \Proc$. A command $\clet{\vx}{\expr}{\cmd}$
declares a local variable $\vx$, initialised by $\expr$, within command $\cmd$.
Further, commands can be defined by \emph{composition},
via a conditional \emph{conditional}, via a \emph{while-loop} construct
or via the \emph{non-deterministic} choice operator $\cnd{\cmd}{\cmdtwo}$, interpreted
in a demonic way.
For ease of presentation, we elide a probabilistic choice command and probabilistic guards (see eg~\cite{KaminskiKMO16,KKMO:ACM:18}), since they do not add to the expressiveness of the language.
In examples, however, we make use of mild syntactic sugaring to simplify
readability, as we did already above.

A \emph{program} $\prog$ is given as a finite sequence of procedure definitions.
Each procedure $\fn$ expects a number of integer-valued arguments
and returns an integer upon termination. The number of arguments of a procedure
$\fn$ is called its \emph{arity} and denoted as $\ar{\fn}$. The body of~$\fn$, denoted as $\bdy{\fn}$,  consists
of a command $\cmd \in \Cmd \app V$ that may refer both to formal parameters, locally and globally declared
variables. Note that since $\fn$ may trigger a sampling instruction, its evaluation
evolves probabilistically, yielding a final value and modifying the global state
with a certain probability.
Formally, a program is a tuple $\prog = (\GVar,\dcl)$ where
$\GVar \subseteq \Var$ is a finite set of \emph{global variables},
and where $\dcl$ is a finite sequence of \emph{procedure declarations} of the form
\begin{equation*}
  \cdef{\fn}{\vx_1,\dots,\vx_{\ar{\fn}}}{\cmd}
  \tpkt
\end{equation*}
%
To avoid notational overhead due to variable shadowing,
we assume that the \emph{formal parameters} $\params{\fn} \defsym \vx_1,\dots,\vx_{\ar{\fn}}$ and global variable
are all pairwise different, and distinct from variables locally bound within $\cmd$.
Using $\alpha$-renaming, this can always been guaranteed.
Throughout the following, we keep the program $\prog = (\GVar,\dcl)$ fixed.
\label{variable_convention}

\paragraph*{Weakest Pre-Expectations Semantics.}

A \emph{memory} (or \emph{state}) over finite variables $V \subseteq \Var$ is a mapping $\mem \in \Mem{V} \defsym V \to \Val$ from variables to integers.
We write $\GMem \defsym \Mem{\GVar}$ for the set of \emph{global memories},
and $\restrict{\mem}{V}$ for its restriction to variables in $V$.
Let $\vx \in V$ be a variable and let $\val \in \Val$.
Then we write $\mem[\vx \mapsto \val]$ for the memory that is as $\mem$
except that $\vx$ is mapped to $\val$.
As short-forms, we write $\memg{\mem}$ for the \emph{global memory} $\restrict{\mem}{\GVar}$,
and dual, $\meml{\mem}$ for the \emph{local memory} $\restrict{\mem}{\Var \setminus \GVar}$.
Let $\Rext$ denote the set of non-negative real numbers extended with $\infty$,
ie.\ $\Rext \defsym \Realpos \cup \{\infty\}$.
A (discrete) \emph{subdistribution} over $A$ is a function $\delta \colon A \to \Realpos$
so that $\sum_{a \in A} (\delta \app a) \leq 1$, and a \emph{distribution},
if $\sum_{a \in A} (\delta \app a) = 1$.
We may write (sub)distributions $\delta$ as $\prms{\delta \app a : a}_{a \in A}$.
The set of all subdistributions over $A$ is denoted by $\Distr \app A$.
We restrict to distributions over countable sets $A$.
The \emph{expectation of} a function $f \colon A \to \Rext$ wrt.\ a distribution
$\delta$ is given by $\E{\delta}{f} \defsym \sum_{a \in A} (\delta \app a) \cdot (f \app a)$.
We suppose that expressions $\expr \in \Expr \app V$,
Boolean expressions $\bexpr \in \BExpr \app V$ and
sampling expressions $\dexpr \in \SExpr \app V$
are equipped with interpretations $\sem{\expr} : \Mem{V} \to \Val$,
$\sem{\bexpr} : \Mem{V} \to \Bool$
and $\sem{\dexpr} : \Mem{V} \to \Distr\ \Val$, respectively.
Functions in $A \to \Rext$ are called \emph{expectation} (over a set~$A$) in the literature
and are usually denoted by $f,g,h$ etc.
We equip expectations and transformers with the point-wise ordering,
that is, $f \leqf g$ if $f \app a \leq g \app a$ for all $a \in A$.
We also extend functions over $A$ point-wise to expectations
and denote these extensions in \texttt{typewriter} font, eg., $f \plusf\ g \defsym \lambda a. f \app a + g \app a$ etc.
In particular, $\zerof = \lambda a. 0$ and $\inftyf = \lambda a. \infty$.
Equipped with this ordering, expectations form an \emph{$\omega$-CPO} \cite{Winskel:93}, whose least element is the constant zero function $\zerof$.

\begin{figure*}
  \caption{Expectation Transformer Semantics of \PWHILE.}
  \label{fig:et}
  \vspace{-3mm}
  \begin{center}
  \begin{alignat*}{3}
    \toprule
    \TOP
    & \mparbox{7mm}{\et[\penv]{\fn} : (\Val \to \GMem \to \Rext) \to (\Val^{\ar{\fn}} \to \GMem \to \Rext)}
    \\[1mm]
    & \et[\penv]{\fn}\app f  &&
    \defsym \lam{\vec{\val}\,\mem}{\et[\penv][f]{\bdy{\fn}}} \app (\lam{\memtwo}{f \app 0 \app \memg{\memtwo}}) \app (\mem \uplus \{\params{\fn}\mapsto \vec{\val}\})
    \\[1.5ex]
    & \mparbox{7mm}{\et[\penv][f]{\cmd} \colon (\Mem{V} \to \Rext) \to (\Mem{V} \to \Rext)}
    \\[1mm]
    & \et[\penv][f]{\cskip} \app g &&
      \defsym g
    \\
    & \et[\penv][f]{\vx \sample \dexpr} \app g &&
    \defsym \lam{\mem}{\mathbb{E}_{\sem{\dexpr} \app \mem} \app (\lam{\val}{f[\vx \mapsto \val]})}
    \\
    & \et[\penv][f]{\vx \sample \ccall{\fn}{\vec{\expr}}} \app g &&
      \defsym \lam{\mem}{\penv \app \fn \app
      (\lam{\val\,\memtwo}{g \app (\meml{\mem}\uplus \memtwo)[\vx \mapsto \val]})
      \app (\sem{\vec{\expr}} \app \mem) \app \memg{\mem}}
    \\
    & \et[\penv][f]{\cret{\expr}} \app g &&
       \defsym \lam{\mem}{f \app (\sem{\expr} \app \mem) \app \memg{\mem}}
    \\
    & \et[\penv][f]{\clet{\vx}{\expr}{\cmd}} \app g &&
       \defsym \lam{\mem} {\et[\penv][f]{\cmd} \app g  \app \mem [\vx \mapsto \sem{\expr}\app\mem] }
    \\
    & \et[\penv][f]{\cmd \sep \cmdtwo} \app g &&
       \defsym \lam{\mem}{\et[\penv][f]{\cmd} \app (\et[\penv][f]{\cmdtwo} \app g) \app \mem}
    \\
    & \et[\penv][f]{\cif{\bexpr}{\cmd}{\cmdtwo}} \app g &&
                                                           \defsym \lam{\mem}{[\sem{\bexpr} \app \mem] \cdot \et[\penv][f]{\cmd} \app g \app \mem + [\sem{\neg\bexpr} \app \mem] \cdot \et[\penv][f]{\cmdtwo}\app g \app \mem}
    \\
    & \et[\penv][f]{\cwhile{\bexpr}{\cmd}} \app g &&
       \defsym \lam{\mem}{\lfp{\lam{G}{\lam{\memtwo}{[\sem{\bexpr} \app \memtwo] \cdot \et[\penv][f]{\cmd} \app G \app \memtwo + [\sem{\neg\bexpr} \app \memtwo] \cdot g \app \memtwo}}} \app \mem}
    \\
    & \et[\penv][f]{\cnd{\cmd}{\cmdtwo}} \app g &&
       \defsym \lam{\mem}{\maxf(\et[\penv][f]{\cmd} \app g,\et[\penv][f]{\cmdtwo} \app g) \app \mem}
    \\
    \bottomrule                                                   
  \end{alignat*}
  \end{center}
\vspace{-5mm}
\end{figure*}

\emph{Expectation transformer} have the shape $F : (A \to \Rext) \to (B \to \Rext)$ (for some sets~$A$ and $B$
respectively). Ordered point-wise these again form an $\omega$-CPO and have thus enough
structure to give a denotational model to programs.
In particular, when $A=B$ the least fixed-point $\lfp{F}$ of $F$ is well-defined
and given by $\sup_n\app (F^n \app \bot_A)$, for $F^n$ the $n$-fold composition of $F$ \cite{Winskel:93}.
Following Dijkstra~\cite{Dijkstra76}, expectation transformers can be seen as giving rise to a (denotational) semantics.%
\footnote{Such a transformer constitutes a standard denotational semantic, where effects are interpreted in the continuation monad $\mathsf{Cont}_{\Rext}({-}) = (({-}) \to \Rext) \to \Rext$, cf.~\cite{ABD:2021}. To observe this, flip the arguments $f$ and $\mem$ in Figure~\ref{fig:et}.}

The transformer $\et{\prog}$ is defined in terms of an \emph{expectation transformer} $\et{\fn}$, $\fn \in \prog$, for \emph{procedures},
which in turn is mutual recursively defined via an \emph{expectation transformer} $\et{\cmd}$ on \emph{commands}, cf.~\Cref{fig:et}.
These transformers are parameterised in a \emph{procedure environment} $\penv$,
of type
\[
  \fn{:}\Proc \to (\Val\times\GMem \to \Rext) \to (\Val^{\ar{\fn}} \times \GMem \to \Rext)
  \tkom
\]
associating an expectation semantic to each $\fn \in \Proc$, which is
used in the case of procedure calls $\vx \sample \ccall{\fn}{\vec{\expr}}$.
As already alluded to, the definitions given in~\Cref{fig:et} are best understood as a denotational semantics in continuation passing style, with post-expectations $f \colon \Val \times \GMem \to \Rext$ and $g : \Mem{V} \to \Rext$ being interpreted as continuations, respectively.
%
In short, the functional
\[
  \et{\fn} : (\Val \times \GMem \to \Rext) \to (\Val^{\ar{\fn}} \times \GMem \to \Rext)
  \tkom
\]
turns a continuation in the return-value
and (possibly modified within the body of $\fn$) global memory, and lifts it to a function
in the formal parameters of $\fn$ and the initial global memory.
To this end, $\et{\fn}$
links formal parameters $\params{\fn}$ to the input values $\vec{\val}$,
and yields to the transformer $\et[\penv][f]{\bdy{\fn}}$ associated to its body.
As denoted, this auxiliary transformer on commands is parameterised by a procedure environment $\penv$,
and continuation $f$ of the procedure $\fn$. The latter is necessary to model evaluation when $\bdy{\fn}$ pre-maturely leaves scope through a return statement. To capture the sitation where evaluation completes without encountering a return, the continuation to $\et[\penv][f]{\bdy{\fn}}$ supplies to $f$ a default return value of zero.

To give some intuitions on the transformer of commands, it is again advisive to think of
$\et[\penv][f]{\cmd} \app g \app \mem$ as running the command $\cmd$ on memory $\mem$, and then proceeding with continuation~$g$. The only exception to this reading lies in the treatment of sampling
instructions.

For a sampling instructions $\vx \sample \dexpr$, $\et[\penv][f]{\cmd} \app g \app \mem$ computes
the expected value of the continuation~$g$ on the distribution of stores obtained by updating $\vx$ with elements sampled from $\dexpr$. (As a simplified and rough intuition, think of the assignment rule in Hoare logics.)
For a procedure call $\vx \sample \ccall{\fn}{\vec{\expr}}$,
we make use of the semantics $\penv \app \fn$ of $\fn$,
which is applied to the evaluated arguments and the current global memory $\memg{\mem}$.
The continuation passed to $\penv \app \fn$ runs the continuation $g$ on the combination of global
memory $\memtwo$ yielded by $\fn$ and the local pre-memory $\meml{\mem}$, with $\vx$ updated by the return value yielded by $\fn$. 
When $\cmd$ is a return statement, the transformer skips continuation~$g$ and jumps
directly to the continuation $f$ defined by the enscoping procedure, supplying the returned value.
For a local variable declarations, $\et[\penv][f]{\clet{\vx}{\expr}{\cmd}} \app g$ implement lexical scope, updating variable $\vx$ in $\mem$ by $\expr$.
Due to the variable convention---as emphasised on page~\pageref{variable_convention}---%
the \emph{local} variable $\vx$ is fresh and thus cannot conflict with any variable in $V$.

\begin{figure*}
\caption{Expectation Transformer Laws.}
\label{fig:et-idents}  
  \begin{tblr}{colspec={X[-1,l,mode=math]X[mode=math,l,m]},rowsep={3pt}}
    \toprule
    \TOP
    \law[idents:mono]{monotonicity}
    & \penv \leq \penvtwo \land f_1 \leqf f_2 \land g_1 \leq g_2 \implies \et[\penv][f_1]{\cmd} \app g_1 \leq \et[\penvtwo][f_2]{\cmd} \app g_2 \\[1mm]
    \law[idents:linearity]{linearity}
    & \et[\penv][\sum_i  \coeff_i \cdot f_i]{\cmd} \app (\sum_i \coeff_i \cdot g_i) = \sum_i \coeff_i \cdot \et[\penv][f_i]{\cmd} \app g_i
    \\
    \law[idents:ui-loop]{loop-invariant}
    & [\neg \bexpr] \cdot g_1 \leqf g_2 \land [\bexpr] \cdot \et[\penv][f]{\cmd} \app g_2 \leq g_2 \implies \et[\penv][f]{\cwhile{\bexpr}{\cmd}} \app g_1 \leq g_2 \\
    \BOT \law[idents:ui-proc]{procedure-invariant}
    & \forall \fn \in \Proc.\ \et[\penvtwo][f]{\fn} \leqf \penvtwo \app \fn \implies \progenv \leq \penvtwo  \\
    \bottomrule
  \end{tblr}
\vspace{-3mm}  
\end{figure*}

The transformer for composed commands is given by the composition of the corresponding transformers.
In the case of conditionals, we use Iverson bracket $\brac{\cdot}$ to interpret Boolean values $\bfalse$ and $\btrue$ as integers $0$ and $1$, respectively. Thus the transformer of conditionals effectively
recurses on one of the two branches of the conditional, depending on the condition $\bexpr$.
The expectation transformer for loops can be seen as (the least fixed-point) satisfying
\[
  \et[\penv][f]{\cwhile{\bexpr}{\cmd}} \app g \app \mem
  =
  \begin{cases}
    \et[\penv][f]{\cmd\sep\cwhile{\bexpr}{\cmd}} \app g \app \mem
     & \text{if $\sem{\bexpr} \app \mem = \btrue$,} \\
    g \app \mem & \text{if $\sem{\bexpr} \app \mem = \bfalse$,}
  \end{cases}
\]
running $\cmd$ once followed by the while loop in case the guard holds;
and calling the continuation $g$ otherwise.
Finally, non-determinism is modelled as the \emph{maximum} of the pre-expectations obtained
from the alternatives. This is motivated by the fact that we are interested in worst-case bounds
(on expected values and costs) and follows the treatment
of non-determinism in the context of program analysis, cf.~\cite{Nielson:1999}.
%
We note that our treatment of non-determinism constitutes a conceptual difference
to the \emph{weakest pre-expectation calculus} of~\cite{MM05}. There the focus is
on (quantitative) program behaviours and thus a lower-bound to the pre-expectation is sougth
which results in choosing the \emph{minimum} of the pre-expectations to handle
non-deterministic choice.

What remains is to set up the procedure environment $\penv$ according to the
declarations in $\prog$.
To this end, we associate the semantics $\et{\prog}$ with the procedure environment $\penv$
that makes each $\fn \in \Proc$ adhere to the semantics $\et[\penv]{\fn}$, that is,
that satisfies the (least) fixed point $\penv \app \fn \app f = \et[\penv]{\fn} \app f$.
More precisely, $\et{\prog} \defsym \lfp{\lam{\Xi}{\lam{\fn}{\et[\Xi]{\fn}}}}$.
Again, this least (higher-order) fixed-point exists as expectation transformers form an $\omega$-CPO.

We emphasise that all the individual transformers are defined mutually, thus permitting mutual recursion on procedure declarations.
In the following, we write $\et{\cdot}$ instead of $\et[\penv][f]{\cdot}$
when $\penv = \et{\prog}$ and caller expectation $f$ is clear from context.

\begin{proposition*}[Expectation Tranformer Laws]{p:laws}
  For any program $\prog$, any procedure environment $\penv$, any command $\cmd$ and any expectations
  $f$, $f_1$, $f_2$, $g$, $g_1$, $g_2$, \dots the laws in Figure~\ref{fig:et-idents} hold.
\end{proposition*}

\section{Inference}
\label{Inference}

In this section we present the key contribution of this work, the development of a \emph{term representation} of the
expectation transformer $\et{\prog}$, see~Figure~\ref{fig:ET}. This 
forms the crucial basis of the automated inference of upper bounds to~$\et{\prog}$, as
implemented in our prototype~\evimp.

To a great extent, the definition of $\ET{\fn}$, $\fn$ a procedure, follows the pattern of the definition of $\et{\fn}$.
Notably, however, it differs in the definition of procedure calls, where we employ the templates $\soa{\termh{\fn}}{\termk{\fn}}$
outlined in Section~\ref{Overview} and the definition of loops, where we replace the fixed-point construction via a suitably constrained template.
Theorem~\ref{t:soundness} verifies that this approximation is sound; in proof we make essential use of the invariant laws depicted in Figure~\ref{fig:et-idents}.
We represent expectations syntactically as terms, denoting linear combinations of norms:
\begin{equation*}
  \Norm{V}{Z} \ni \norm \bnfdef \bfn{\bexpr}{\expr} \quad (norms) \qquad
  \Term{V}{Z} \ni \term,\termtwo,\termthree \bnfdef \sum_i \coeff_i \cdot \norm_i \quad (terms)
\end{equation*}
$\expr \in \Expr \app (V \uplus Z)$ denotes an arbitrary expression over \emph{program} variables $V$ and
\emph{logical} variables $Z = \{\lv,\lvtwo,\dots\}$.
$\bexpr \in \BExpr\app (V \uplus Z)$ denotes a Boolean expression over program and logical variables.
Coefficients $\coeff$ denote terms yielding non-negative real numbers.

We require that for any norm $\norm$, $\expr$ is non-negative, whenever $\bexpr$ holds.
A norm abstracts an expression over a program variable as a non-negative real number.
For instance, $\max(\vx,0) = \bracf{\vx \geq 0} \cdot \vx$, or $\bracf{\vx \geq \vy} \cdot (\vx - \vy)$ gives
the distance from $\vx$ to $\vy$. For brevity, we set $\normf{\vx} \defsym \bracf{\vx \geq 0} \cdot \vx$.
Following the semantics, we set $[\bexpr] \cdot (\sum_i \coeff_i \cdot (\bfn{\bexpr_i}{\expr_i})) = \sum_i \coeff_i \cdot (\bfn{\bexpr \land \bexpr_i}{\expr_i})$.
Note that, since norms are non-negative, a term $\term \in \Term{V}{\varnothing}$ can be interpreted as an expectation
$\sem{\term} : \Mem{V} \to \Rext$ over the state-space.
Let $\dexpr \in \SExpr$ denote a sampling instruction and let $\term \in \Term{V}{\varnothing}$,
then $\ETerm{\vx}{\dexpr}{\term}$ denotes the term representation of the expectation
$\E{{\sem{\dexpr} \app \mem}}{\sem{\term[\vx \mapsto \val]}}$
of the continuation $\sem{\term[\vx \mapsto \val]}$ wrt.\ the distribution $\sem{\dexpr}$ applied to the current memory $\mem$.

\newcommand{\DGREEN}[1]{{\color{green!50!black}{#1}}}
\begin{figure*}
  \caption{Term Representation of Expectation Transformer Semantics.}
  \label{fig:ET}
  \vspace{-3mm}
  \centering
  \begin{alignat*}{3}
    \toprule
    \TOP
    &  \mparbox{7mm}{\ET[\penv]{\fn} : \Term{\GVar}{\{\lvr\}} \to \Term{\GVar}{\{\lvas{\fn}}\}}
    \\[1mm]
    & \ET[\penv]{\fn}\app \termtwo &&
    \defsym (\ET[\penv][\termtwo]{\bdy{\fn}} \app \termtwo [\lvr \mapsto 0]) [\params{\fn} \mapsto \lvas{\fn}]
    &
    \\[1.5ex]
    &  \mparbox{7mm}{\ET[\penv][\termtwo]{\cmd} : \Term{V}{\{\lvs[\fn]\}} \to \Term{V}{\{\lvs[\fn]\}}}
    \\[1mm]
    & \ET[\penv][\termtwo]{\cskip} \app \term && \defsym \term
    \\
    & \ET[\penv][\termtwo]{\vx \sample \dexpr} \app \term &&
      \defsym \ETerm{\vx}{\dexpr}{\term}
    \\
    & \ET[\penv][\termtwo]{\vx \sample \ccall{\fntwo}{\vec{\expr}}} \app \term
      && \defsym \termh{\fntwo}[\vec{\lva}_{\fntwo} \mapsto \vec{\expr}, {\lvs[\fntwo]} \mapsto \vec{\termthree}]
      \text{ \emph{where} } \vec{\termthree} \in \Term{(V \setminus \GVar)}{\{\lvs[\fn]\}}
    \\
    &&&& \mparbox[r]{0mm}{
         \DGREEN{\ctx[\fntwo] \vdash \ctx[\fntwo] [{\lvs[\fntwo]} \mapsto \vec{\termthree}]}; \quad 
         \DGREEN{\ctx[\fntwo] \vdash \term[\vx \mapsto \lvr] \leq \termk{\fntwo}[{\lvs[\fntwo]} \mapsto \vec{\termthree}]}}
    \\
    & \ET[\penv][\termtwo]{\cret{\expr}} \app \term && \defsym \termtwo[\lvr \mapsto \expr]
    \\      
    & \ET[\penv][\termtwo]{\clet{\vx}{\expr}{\cmd}} \app \term
    && \defsym (\ET[\penv][\termtwo]{\cmd} \app \term)[x \mapsto \expr]
    \\    
    & \ET[\penv][\termtwo]{\cmd \sep \cmdtwo} \app \term
    && \defsym \ET[\aenv][\termtwo]{\cmd} \app (\ET[\aenv][\termtwo]{\cmdtwo} \app \term)
    \\
    & \ET[\penv][\termtwo]{\cif{\bexpr}{\cmd}{\cmdtwo}} \app \term
    && \defsym [\bexpr] \cdot \ET[\aenv][\termtwo]{\cmd} \app \term + [\neg\bexpr] \cdot \ET[\aenv][\termtwo]{\cmdtwo}\app \term
    \\
    & \ET[\penv][\termtwo]{\cwhile{\bexpr}{\cmd}} \app \term
      && \defsym \termthree
    & \mparbox[r]{0mm}{\DGREEN{\bexpr \vdash \ET[\aenv][\termtwo]{\cmd} \app \termthree \leq \termthree}; \quad \DGREEN{\neg\bexpr \vdash \term \leq \termthree}}\\
    & \ET[\penv][\termtwo]{\cnd{\cmd}{\cmdtwo}} \app \term
      && \defsym \termthree
                                   & \mparbox[r]{0mm}{\DGREEN{\vdash \ET[\aenv][\termtwo]{\cmd} \app \term \leq \termthree}; \quad \DGREEN{\vdash \ET[\aenv][\termtwo]{\cmdtwo} \app \term \leq \termthree}}
    \\
  \bottomrule
  \end{alignat*}
\vspace{-5mm}
\end{figure*}

The definition of $\ET{\cdot}$ (see~\Cref{fig:ET}) is best understand as a syntactic representation of the
expectation transformer $\et{\cdot}$. It translates this denotational semantics into first-order
constrains, susceptible to automation.
Apart from returning a term, representing a pre-expectation,
it generates a set of \emph{side-conditions} of the form
\DGREEN{$\ctx \vdash \termtwo \leq \term$}. Such a constraint is \emph{valid}, if for all
logical variables $\vec{\lv}$ occurring in the expressions $\bexpr$, $\termtwo$ and $\term$, respectively, we have 
$\sem{\bexpr[\vec{\lv} \ass \vals]} \entails \sem{\termtwo[\vec{\lv} \ass \vals]} \leq \sem{\term[\vec{\lv} \ass \vals]}$ for all $\vals$. In this reformulation, the procedure environment $\penv$ is kept implicit.
The semantics of each $\fn \in \Proc$ is thus representable as a pair of terms
$\termk{\fn} \in \Term{\GVar \uplus \{\lvr\}}{Z}$,
$\soa{\termh{\fn}}{\termk{\fn}}$, with $\termh{\fn} \in \Term{\GVar \uplus \{\lvas{\fn}\}}{Z}$ and
where $\lvas{\fn} = \lva_1,\dots,\lva_{\ar{\fn}}$ and $\lvr$ are dedicated variables,
referring to the formal parameters and the return value of $\fn$.

For each $\fn$, the terms $\termh{\fn}$ and $\termk{\fn}$ can be understood as families of terms, parameterised in
the substitution of logical variables. To wit, let $\theta$ denote an arbitrary substitution of
logical variables for values $\vals$, then $\sem{\termh{\fn}\theta}$ and $\sem{\termk{\fn}\theta}$ denote
pre- and post-expectations $h_{\vals} \colon \Val^n \times \GMem \to \Rext$
and $k_{\vals} \colon \Val \times \GMem \to \Rext$, respectively. Conceptually, the logical variables
model \emph{resource parametricity}~\cite{Hoffmann11}, which is required to model the call-stack in recursive
calls suitably. 
To improve upon the expressiveness of templates, we require that the terms
$\termh{\fn}$ and $\termk{\fn}$ are non-negative only under an
\emph{associated constraint} $\ctx[\fn]$ on logical variables.

For instance, in \Cref{Overview}, we implicitly used the logical context $\ctx[\pwhile!balls!] = (0 \leq \lv)$
to ensure that templates $\termh{\pwhile!balls!} = c_{0} + c_{1} \cdot \normf{\lva} + c_{2} \cdot \lv$
and $\termk{\pwhile!balls!} = \lvr + \lv$ are non-negative.
To ensure that all the pairs $\soa{\termh{\fn}}{\termk{\fn}}_{\fn \in \Prog}$ adhere to the semantics of $\prog$,
wrt.\ to their associated contexts $\ctx[\fn]$, we finally require (for all $\fn \in \prog$)
\begin{equation}
  \label{eq:main-constraints}
  \DGREEN{\ctx[\fn] \vdash (\ET{\fn} \app \termk{\fn})[\params{\fn} \mapsto \vec{\varg}] \leqf \termh{\fn} \tpkt}
\end{equation}
The left-hand side of the inequality in~\eqref{eq:main-constraints} may reference the pair $\soa{\termh{\fn}}{\termk{\fn}}$,
namely when $\fn$ calls itself recursively. As we have already seen in the informal description in Section~\ref{Overview} this
recourse may be performed for an \emph{instantiation} $[\lvs[\fn] \mapsto \vec{\termthree}]$ of logical
variables. In Section~\ref{Implementation} we detail how this instance is chosen.
Under the intended meaning of $\soa{\termh{\fn}}{\termk{\fn}}$,
any application of $\soa{\termh{\fn}}{\termk{\fn}}$ to an expectation $\term$
(an alternation of $\termk{\fn}$)
can safely be replaced by $\termh{\fn}$, provided the passed expectation $\term$
is bounded again from above by $\termk{\fn}$, viz, the corresponding
constraint $\ctx[\fntwo] \vdash \term[\vx \mapsto \lvr] \leq \termk{\fntwo}[{\lvs[\fntwo]} \mapsto \vec{\termthree}]$
in Figure~\ref{fig:ET}.

\begin{theorem*}[Soundness Theorem]{t:soundness}
  If for all $\fn \in \prog$ the constraint~\eqref{eq:main-constraints} is fulfilled and all side-conditions in \Cref{fig:ET} are met, then for all $\fn \in \prog$, $\et{\fn} \leqf \sem{\ET{\fn}}$,
  that is, the inference algorithm is sound.
\end{theorem*}

\begin{figure*}
\caption{Studies in Expected Value Analysis}
\label{fig:motivation2}  
\smallskip
\centering
\begin{subfigure}[b]{0.37\textwidth}
\begin{lstlisting}[style=pwhile,style=framed,mathescape]
def biased_coin($x_{1}$,$x_{2}$):
if ($\ast$) {
 if (Bernoulli($\sfrac{1}{2}$)) {
   $x_1 \ass 2 \cdot  x_1$; $x_{2} \ass 2 \cdot x_{2}$;
   if ($x_{2}+1 \leq x_1$) {$x_{2} \ass x_{2}+1$}
   else {
     if ($x_{2} + \sfrac{1}{2} \leq x_1$) {$x_1 \ass 1$; $x_{2} \ass 0$}
     else {$x_1 \ass 0$; $x_{2} \ass 0$}}}};
 return $x_1$
\end{lstlisting}
\caption{Generating a biased coin~\cite{WangHR18}.}
\label{lst:biased}
\end{subfigure}
\hspace{1ex}\hfill
\begin{subfigure}[b]{0.25\textwidth}
\begin{lstlisting}[style=pwhile,style=framed,mathescape]
def binomial_update(N):
  var x $\ass$ 0;
  var n $\ass$ 0;
  while (n < N) {
    x $\ass$ x + Bernoulli($\sfrac{1}{2}$);
    n $\ass$ n + 1
  };
  return x

\end{lstlisting}
\caption{Binomial update~\cite{KatoenMMM10}.}
\label{lst:binomial}
\end{subfigure}
\hspace{1ex}\hfill
\begin{subfigure}[b]{0.32\textwidth}
\begin{lstlisting}[style=pwhile,style=framed,mathescape]
def hire(n):
   var hires $\ass$ 0;
   if (n > 0) {
     hires $\ass$ hire(n-1);
     hires $\ass$ hires+Bernoulli($\sfrac{1}{n}$)
   };
   return hires


\end{lstlisting}
\caption{Hire a new assistant~\cite{Cormen:2009}.}
\label{lst:hire}
\end{subfigure}
\label{fig:inference}
\vspace{-3mm}
\end{figure*}

In the remainder of the section, we detail the definition of the
term representation~\Cref{fig:ET}, in particular on the
delineated constraints. To this avail, we consider its working wrt.\ three prototypical benchmark
examples, depicted in Figure~\ref{fig:motivation2}.

\paragraph{Templating approach.}

In the inference of upper invariants, we follow the \emph{templating approach} standard in the literature
(see eg.~\cite{AMS20,WangKH20,LMZ:2022}) in which the functions to be synthesised---in our case a
closed form of $\ET{\prog}$---are given as linear combination $\sum_i c_i \cdot b_i$ of pre-determined
\emph{base functions}~$b_i$, with variable coefficient $c_i$.
We emphasise that base functions can be \emph{linear} or \emph{non-linear}.
Concretely, for straight-line commands $\cmd$ this is captured by the definition of
the term representation $\ET{\cmd}$ in the context of the 
continuation $\term$. In Section~\ref{Overview} we have seen an informal account of this
recipe. As a slightly more involved example, consider procedure \pwhile!biased_coin! depicted in Listing~\ref{fig:motivation2}(\subref{lst:biased}).%
\footnote{%
  Listing~\ref{fig:motivation2}(\subref{lst:biased}) is taken from Wang et al.~\cite{WangHR18} and---making use
  of a non-deterministic abstraction---constitutes a
  variant of an example considered by Katoen et al.~\cite{KatoenMMM10}.
  The latter example uses a stream of fair coin flips to generate a (single) biased coin.}
The procedure incorporates a \emph{non-deterministic choice}, denoted
by the conditional with unspecified guard $(\ast)$. 

We focus on the \emph{non-deterministic choice} and the \emph{sampling instruction}
incorporated.
Wrt.~$\cnd{\cmd}{\cmdtwo}$ the term representation
eludes an explicit representation of the maximum function $\maxf$. Instead non-determinism
is resolved by asserting the constraints
\begin{inparaenum}[(i)]
\item \DGREEN{$\vdash \ET[\aenv][\termtwo]{\cmd} \app \term \leq \termthree$} and
\item \DGREEN{$\vdash \ET[\aenv][\termtwo]{\cmdtwo} \app \term \leq \termthree$}, respectively.
\end{inparaenum}
Wrt.\ sampling instructions, $\ET[\penv][\termtwo]{\vx \sample \dexpr} \app \term$ is defined as the term
$\ETerm{\vx}{\dexpr}{\term}$, representing the computation of the expected value
of the continuation $\term$ symbolically on the distribution of the memory obtained
by sampling elements according to the instruction $\dexpr$.

To illustrate, let $\cmd$ denote the body of the procedure \pwhile!biased_coin!. By default, we choose the return value $x_{1}$ as continuation, which we abstracted by the norm $\normf{x_{1}}$.
Restricting to the non-trivial branch of the non-deterministic choice, we obtain by
symbolic calculation that
$\ET{\cmd} \app \normf{x_1} = \normf{x_1} + \bracf{x_1 \geq  \sfrac{1}{2} + x_2}$.
Solely by symbolic execution, w obtain the constraint
%
\begin{align*}
  \entails \normf{x_1} + \bracf{x_1 \geq  \sfrac{1}{2} + x_2}
  & \leqf c_0 \cdot \bracf{x_1 \geq  \sfrac{1}{2} + x_2} \cdot \normf{x_1}
  + c_1 \cdot \bracf{x_2 \geq x_1} \cdot \normf{x_1} + \sfrac{1}{2} \cdot c_2 \cdot \bracf{x_1 \geq  \sfrac{1}{2} + x_2}
    \tkom
\end{align*}
solvable with $c_0 = c_1 = c_2 = 1$, yielding $\normf{x_1} + \bracf{x_1 \geq  \sfrac{1}{2} + x_2}$
as the desired expected return value.

\paragraph{Loop programs.}

Considering loops, we employ the
\emph{loop-invariant law} to derive a closed form, cf.~Figure~\ref{fig:et-idents}.
Conclusively, for a loop statement $\cwhile{\bexpr}{\cmd}$, $\ET[\penv][\termtwo]{\cwhile{\bexpr}{\cmd}} \app \term$
asserts the constraints
\begin{inparaenum}[(i)]
  \item \DGREEN{$\bexpr \vdash \ET[\aenv][\termtwo]{\cmd} \app \termthree \leqf \termthree$} and
  \item \DGREEN{$\neg\bexpr \vdash \term \leqf \termthree$}, respectively.
\end{inparaenum}
Ie.\ $\termthree$ represents an upper bound $I_{\term}$, parameterised in the
post-expectation $\term$. Due to its reminiscence with a loops invariant, the
function~$I_{\term}$ is called an \emph{upper invariant} in the literature \cite{KKMO:ACM:18}.

To illustrate, consider the procedure \pwhile!binomial_update! from
Figure~\ref{fig:motivation2}(\subref{lst:binomial}).%
\footnote{The code sets variable $x$ to a value between $0$ and $N$, following
a binomial distribution, cf.~\cite{KatoenMMM10}.}
Again we approximate the return value by the norm $\normf{x}$. The above constraints
on the term $\termthree$ induce the following constraints on an upper invariant $I_{\normf{x}}$
\begin{inparaenum}[(i)]
\item $N \leq n \entails \normf{x} \leqf I_{\normf{x}}$ and 
\item $n < N \entails \sfrac{1}{2} \cdot \et{x \sample x + 1; n \sample n +1} \app I_{\normf{x}} \leqf I_{\normf{x}}$,
\end{inparaenum}
respectively. Making use of linear template based on base functions $\onef$, $\normf{x}$
and $\normf{N-n}$ this can be instantiated as the following constraints.
\begin{align*}
  N \leq n & \entails \normf{x} \leqf c_0 + c_1 \cdot \normf{N-n} + c_2 \cdot \normf{x}
  \\
  n < N & \entails c_0 + c_1 \cdot \normf{N-(n+1)} + \sfrac{1}{2} \cdot c_2 \cdot \bigl(\normf{x} +  \normf{x+1}\bigr)
          \leqf c_0  + c_1 \cdot \normf{N-n} + c_2 \cdot \normf{x}
          \tkom
\end{align*}
solvable with $c_0 = c_1 = 0$ and $c_2 = \sfrac{1}{2}$, yielding the \emph{quantitative invariant} $\sfrac{1}{2} \cdot x$, that is, the derived upper invariant is optimal, cf.~\cite{KatoenMMM10}.

\paragraph{Recursive procedures.}

Since we represent $\et{\fn}$ through the term representation $\ET{\fn}$, a similar
approach can be suited to recursively defined procedures $\cdef{\fn}{\vxs}{\bdy{\fn}}$. Instead
of the \emph{loop-invariant law}, we implicitly employ the law on \emph{procedure-invariants} to derive a closed form,
though (see Figure~\ref{fig:et-idents}).
Thus, the definition of $\ET[\penv][\termtwo]{\vx \sample \ccall{\fntwo}{\vec{\expr}}} \app \term$ asserts
the constraints
\begin{inparaenum}[(i)]
\item \DGREEN{$\ctx[\fntwo] \vdash \ctx[\fntwo] [{\lvs[\fntwo]} \mapsto \vec{\termthree}]$} and
\item \DGREEN{$\ctx[\fntwo] \vdash \term[\vx \mapsto \lvr] \leqf \termk{\fntwo}[{\lvs[\fntwo]} \mapsto \vec{\termthree}]$},
  respectively.
\end{inparaenum}
Here, the latter constraint guarantees that the continuation $\term[\vx \mapsto \lvr]$ of the procedure
call is bounded by an instance of the bounding function $\termk{\fntwo}$. On the other hand
the first constraint guarantees that the substitution $[{\lvs[\fntwo]} \mapsto \vec{\termthree}]$ of logical variables
employed is properly represented in the context information.
Consider procedure \pwhile!hire! depicted in Listing~\ref{fig:motivation2}(\subref{lst:hire}).%
\footnote{The procedure represents a hiring process, encoded as probabilistic program,
  cf.~\cite{Cormen:2009}.}
Following the recipe employed for the procedure \pwhile!balls!, we generate the templates
\begin{inparaenum}[(i)]
\item $\termk{\pwhile!hire!} \defsym \normf{\lvr} + \lv$ and
\item $\termh{\pwhile!hire!} \defsym c_{0} + c_{1} \cdot \normf{\lvar{n}} + c_{2} \cdot \lv$
\end{inparaenum}
where $\lvar{n}$ refers to the parameter taken by \pwhile!hire!, and $\lv$ is a logical
variable subject to the constraint $\ctx[\pwhile!hire!] \defsym (0 \leq \lv)$.
This results in the following three constraints
\begin{align*}
  0 < \lvar{n},\ctx[\pwhile!hire!] & \entails \sfrac{\normf{1 + \lvr}}{\lvar{n}} + (1 - \sfrac{1}{\lvar{n}}){\normf{\lvr}} + \lv \leqf \normf{\lvr} + (d_{0} + d_{1} \cdot \lv) 
  \\
  \ctx[\pwhile!hire!] & \entails 0 \leqf d_{0} + d_{1} \cdot \lv 
  \\
  \ctx[\pwhile!hire!] &  \entails \bracf{0 < \lvar{n}} \cdot (c_{0} + c_{1} \cdot \normf{\lvar{n}-1} + c_{2}\cdot(d_{0} + d_{1} \cdot \lv)) + \bracf{\lvar{n} \leq 0} \cdot \lv \leqf c_{0} + c_{1}\cdot\normf{\lvar{n}} + c_{2}\cdot \lv 
    \tkom
\end{align*}
solvable with $c_0 = 0$, $c_1 = c_2 = d_0 = d_1 = 1$.
Here, the additional premise $0 < \lvar{n}$
in the first constraint stems from a (forward) analysis of the conditional governing the call to \pwhile!hire!.
This yields $\normf{\lvar{n}}$ as (sub-optimal) upper bound to expected number of hires in \pwhile!hire!.
Note that the expected value for the number of expected hires is given
as harmonic number $H_k \in \Theta(\log(n))$. So linear is the best we can do with the templates provided in our prototype implementation~\evimp.

\section{Automation}
\label{Implementation}

We have implemented the outlined procedures, proven sound in \Cref{t:soundness}, in a prototype implementation, dubbed \evimp.%
\footnote{Our prototype implementation employs libraries of the open-source tool,
  freely available at \url{https://gitlab.inria.fr/mavanzin/ecoimp}.
  In particular, we rely on its auxiliary functionality,
  such as program parsers etc., but also on its underlying constraint solver.}
Our tool \evimp\ estimates upper-bounds on the expected (normalised) return value of procedures $\fn$,
as a function of the inputs, that is, \evimp\ computes an upper-bound to $\et{\fn} \app (\lam{v \_}{\normf{v}})$.
We note that the restriction to measurements of the expected return value does not constitute a real restriction but rather constitutes a slight design simplification, as long as the estimated continuation $f$ is representable as a return expression~$\expr$.
The term representation from \Cref{fig:ET} forms the basis of our prototype's implementation~\evimp.
Below, we highlight the main design choices that lead us from the term representation to a concrete algorithm
and provide ample experimental evidence of the algorithmic
expressivity of our prototype implementation.

\paragraph*{Assignments.}

Our implementation supports sampling from \emph{finite}, \emph{discrete} distributions $\dexpr$,
assigning probabilities $p_{i}$ to values $e_{i}$ ($0 \leqslant i \leqslant k$ for constant $k$), where probabilities  $p_{i}$
are expressed as rationals and values $e_{i}$ as integer expressions.
Both, probabilities and expressions can possibly depend on the current memory,
thus our tool natively supports \emph{dynamic sampling instructions} such as in
\pwhile!Bernoulli($\sfrac{1}{n}$)!,
used for example in our rendering of the textbook example Listing~\ref{fig:motivation2}(\subref{lst:hire}).
Note that the probabilities depend on the value of $n$ taken in the memory under which the distribution is evaluated.
Such dynamic distributions are also required to represent our variant of the Coupon Collector's problem---procedure \pwhile!every!, cf.~Listing~\ref{fig:motivation}(\subref{lst:every}).
Overall, \evimp\ natively implements in this way a variety of standard distributions, in particular it supports sampling
\emph{uniform distributions with bounded support}, \emph{Bernoulli},
\emph{binomial} and \emph{hypergeometric} distributions. Note, that through recursion more involved
distributions can be build, with unbounded and dynamic support.
Practical extensions, eg. further support for dynamic sampling are (in our opinion) engineering tasks that do not require novel theoretical insights.

\paragraph*{Template selection and instantiation.}

The overall approach rests on templating to over-approximate the behaviour of procedures and loops.
As indicated earlier, we describe these templates via linear combinations $\sum_i c_i \cdot \norm_{i}$ of pre-determined
\emph{norms}~$\norm_{i}$ and undetermined coefficients $c_i$.
To determine these templates, our implementation selects a set of candidate \emph{base functions} from post-expectations, program invariants (determined through a simple forward-analysis)
and loop-guards,
loosely following the heuristics of \citet{SZV:2016} and \citet{AMS20}. Specifically, an established invariant or guard $x \leq y$
gives rise to the base function $\normf{x - y}$, modelling the (non-negative) distance between $x$ and $y$. This
heuristic is extended from variables to expressions, and to arbitrary Boolean formulas, and turns out to work well in practice.
These base functions are then combined to an overall \emph{linear}, or \emph{simple-mixed} template, cf.~\citet{contejean:2005}. 
of base functions. For instance, $c_{1} \normf{x} + c_{2} \normf{y} + c_{3} x^{2} + c_{4} y^{2} + c_{5} [x \geq 0 \land y \geq 0](x\cdot y) + c_{6}$
is a simple-mixed template over base functions $\normf{x}$ and $\normf{y}$.
In a similar spirit, we make use of templates for the instantiations $\vec{\termthree} \in \Term{(V \setminus \GVar)}{\{\lvs[\fn]\}}$,
as linear functions in the local and logical variables, in the same vain as we have already done when presenting examples.
Our prototype implements caching and backtracking to test different templates when operating in a modular setting (see the paragraph on modularity below). In particular, non-linear base functions are employed only if the linear ones fail. Here, we follow similar approaches in the literature, cf.~\cite{WangHR18,AMS20,WangKH20,LMZ:2021,LMZ:2022}.

\paragraph*{Constraint solving.}

Evaluation of $\ET{\fn}$ for a procedure $\fn$ results in a set of constraints, whose solution is then used to assess bounding functions.
In effect, these constraints are inequalities over real-valued polynomial expressions over unknown coefficients, enriched with conditionals (through Iverson's bracket). We emphasise that via our definition of $\ET{\cdot}$, the computational very complicated problem of computing the pre-expectation $\et{\fn} \app f$ for a given continuation $f$ (finding solutions to second-order fixed point equations to deal with general recursion) has been transformed into a problem suitable to be expressed in constraints, susceptible to automation. In particular this permit us to use the constraint
solver implemented within \citet{AMS20} tool \ecoimp\ to reason about such constraints.
In brief, the solver resolves conditionals through case analysis, resulting in a set of equivalent constraints over unconditional polynomials,
and then makes essential use of Handelman's theorem~\citep{Handelman88} to turn these into constraints over undetermined coefficients.
This in turn enables the use of off-the-shelf SMT solver supporting \texttt{QF\_NRA}, in our case Z3.

\paragraph*{Improving upon modularity of the analysis.}

As in many denotational semantics, the expectation transformer $\et{\cdot}$ is compositional, which facilitates reasoning.
Unfortunately, this compositionality does not directly give rise to modularity of inference. Indeed, in the
context of our expected result value analysis, such a modular inference is unsound in general.
Recursive procedures and loops cause cyclic dependencies that hinder modularity.
In our setting, these cyclic dependencies are reflected within the constraints generated by $\ET{\cdot}$.
Templates assigned to nested loops, or potentially mutually recursive procedures, are thus defined through cyclic constraints.
As a result, constraints cannot be solved in isolation, effectively rendering the approach described so far a whole program
analysis.

In some cases, however, stratification present in the program can be exploited to run our machinery in an iterative way,
thereby greatly improving upon modularity, and in consequence improving upon the performance of the overall implementation.
%
To illustrate, eg.\ one case where our implementation departs from the definition of $\ET{\cdot}$ lies within the analysis of nested loops.
Here our implementation follows ideas originally proposed by~\citet{AMS20}.
In particular, pre-expectations of nested loops can be determined in isolation.
In essence, this modular treatment of loops depends on the linear shape of templates, and the linearity of expectations law (see~\Cref{fig:et-idents}).
While the approach can in general not be lifted to our setting with procedure calls, for a sizeable class of loops---for example those that do not contain recursive calls---the approach extends to our setting.
In this specialised case the constraints imposed by the treatment of the loop
are free of unknowns expect those mentioned in the template $\termthree$ over-approximating the expectation of the loop.
Thus, the added two constraints can be solved independently, and consequently,
a concrete upper-bound rather than a template term $\termthree$ can be substituted for $\ET[\penv][\termtwo]{\cwhile{\bexpr}{\cmd}} \app \term$.
Soundness of this specialisation follows by a straightforward adaption of~\cite[Theorem~7.5]{AMS20}.
These efforts yield that our prototype implementation can analyse the entirety of our $53$ benchmarks examples in around 5~minutes on a standard desktop.
Accepting a slight deterioration of strength (loosing one example), this benchmark can be handled in less than 5~seconds.

\paragraph*{Evaluation.}

\newcommand{\fail}{---}
\newcommand{\notsupported}{\textsf{not supported}}
\newcommand{\tdiff}[2]{\ifthenelse{\equal{#2}{}}{---}{(\fpeval{round(#2/#1,0,0)})}}
\newcommand{\experiment}[7]{#1 & #2 & #4 & #6 & \tm{#3} & \tm{#5} & \tdiff{#3}{#5} & \tm{#7} & \tdiff{#3}{#7}\\}
\newcommand{\tm}[1]{\ifthenelse{\equal{#1}{}}{}{$#1$}}
\newcommand{\benchmark}[1]{
  \multicolumn{3}{@{\,}l}{\textbf{#1}}\\
  \midrule
}
\newenvironment{experimenttable}[1]{%
    \begin{table}
      \caption{Automatically Derived Bounds on the Expected Value via our Prototype~\evimp.}\label{#1}\vspace{-3mm}
      \footnotesize
      \begin{tabular}{@{\,}>{\ttfamily}l @{\hspace{1ex}} l @{\hspace{1ex}} l @{\hspace{1ex}} c @{}}
        \toprule
        \textnormal{\textbf{Program}} & \mparbox{9ex}{\textbf{Return Value}} & \multicolumn{1}{c}{\textbf{Inferred Bound}} & \textbf{Time}
        \\
        \phantom{uniform-dist-100}
        & \phantom{$number$}
                                                         & \phantom{$10 \cdot \normf{price + 1} + 5 \cdot \normf{price -1} + [price \geq 1] \cdot \sfrac{25}{6} \cdot (price^2 - 2 \cdot price + 1)$} & \textbf{(sec)}
        \\%
        }{%
        \bottomrule
      \end{tabular}
    \end{table}
}

To the best our knowledge there is currently no tool providing a \emph{fully automated} expected value
analysis of programs available, regardless whether the considered programming language is
imperative or not, admits recursive programs or not.
On the other hand, there is ample work on (automated) generation of \emph{quantitative invariants}, see eg.~\cite{MM05,KatoenMMM10,ChakarovS14,WangHR18,BaoTPHR:2022}, employing a variety of techniques. Furthermore, there is a large body work on \emph{expected cost analysis},
see eg.~\cite{KKMO:ACM:18,NgoCH18,WFGCQS:PLDI:19,AMS20}. Thus we have chosen examples from these seminal works
as basis of the developed benchmark suite. In addition, we have added a number of examples of our own, detailed
below. In sum this amounts to a test-suite of $53$ examples.
The results of these evaluations are given in Tables~\ref{tab:evaluation:1} and~\ref{tab:evaluation:2}.%
\footnote{For ease of readability and comparison to related works, we have performed basic simplification on the presentations of the bounds that have not yet been incorporated into our prototype implementation~\evimp.}
In short, we can handle $\sfrac{49}{53}$ of the benchmark suite, without any recourse to user interaction.

\paragraph*{Selection and evaluation results on examples on invariant generation.}

We haven chosen challenging examples from \cite{KatoenMMM10,ChakarovS14}, detailed in benchmarks~(a) and~(b)
in Table~\ref{tab:evaluation:1}. Our tool can handle all of these examples,
with the exception of \texttt{uniform-dist}, where
we can only handle a restrictive instance (denoted as \texttt{uniform-dist-100} in the benchmark).
We also note that the expectations employed in \cite{KatoenMMM10} are more sophisticated than ours. On the
other hand only semi-automation is achieved.
Further, we consider examples from~\cite{WangHR18,BaoTPHR:2022},
establishing fully automated methodologies. Wang et al.\ employ a templating approach similar to ours,
while the very recent work by Bao et al.\ employs a conceptually highly interesting learning approach.%
\footnote{These artifacts are freely available; see~\url{https://dl.acm.org/do/10.1145/3211994/full/} and~\url{https://zenodo.org/record/6526004\#.Y1JMA35By5M}, respectively.}
The results are described in Table~\ref{tab:evaluation:1}~(c), (d) as well as in Table~\ref{tab:evaluation:2}~(d),
respectively. To suit to the expressivity of \evimp, we instantiate the probabilities by constant once
in the majority of the benchmarks in~(d). Variable probabilities are not (yet) expressible as upper invariants in our prototype.

We can handle all expect one example from these benchmarks. 
In the one example \texttt{eg} (benchmark~(c)) that we cannot handle the templates chosen are not precise
enough, which is explained by the more liberal construction we use to handle
non-determinism in the definition of $\ET{\cdot}$.
Wrt.\ precision, the bounds generated by~\evimp\ are often as precise as those generated by the tools from
\citet{WangHR18} and \citet{BaoTPHR:2022}. 
Wrt.\ to speed, our prototype implementation \evimp\ is typically able to handle the benchmark
examples in milliseconds. As already mentioned the whole testsuite can be handled
in around 5~minutes on a standard desktop.

\paragraph*{Selection and evaluation results on examples on expected costs.}

As argued by \citet{KKMO:ACM:18} an expected value analysis is in general unsound for expected cost analysis. However, if the program under consideration is \emph{almost-surely terminating} (\emph{AST} for
short), then an expected cost analysis can be recovered from an expected value analysis by counter instrumentation.
In this context, an expected value analysis constitutes a strict extension of an expected cost analysis. Note that all benchmarks considered in Table~\ref{tab:evaluation:2}~(f) and~(g) are AST (even \emph{positive almost-sure terminating}).

We have incorporated recursive examples from~\citet{KKMO:ACM:18}, suited to our possibilities, cf.~Table~\ref{tab:evaluation:2}~(e). While procedure \texttt{geo} can be handled instantly with \evimp, the growth rate of the
(expected) value of \texttt{faulty\_factorial} is non-polynomial. We thus suited a variant---dubbed
\texttt{faulty\_sum}---to our benchmarks replacing the multiplication by a sum, thus
featuring polynomial growth but similar algorithmic complexity.
In addition, we considered challenging---newly introduced examples from~\citet{AMS20}---if expressible in our
prototype. Wrt.\ both benchmarks, we adapted the original expected cost analysis to an expected return value analysis,
as proposed above. The evaluation results are given in Table~\ref{tab:evaluation:2}~(f) and~(g).
Remarkably, we can handle both recursive examples from Kaminski et al.\ and
also make significant inroad into challenging non-linear example from~\cite{AMS20}. Unsurprisingly,
we cannot handle all the selected benchmarks due to the greater generality of our prototype~\evimp. 

For example, the crucial motivating example by Avanzini et al.---\emph{Coupon Collector}---cannot
(yet) be handled by our prototype implementation~\evimp, as support for dynamic uniform distributions is lacking.
Our central contributions is a proof-of-concept of an automated (expected) value analysis of recursive programs.
In addition we seek general (and efficient) applicability. To this avail, we incorporated the
modular framework of~\citet{AMS20}, where possible without violating soundness.
We believe that these insights provide a sensible basis for future work. Extensions, as for example required to
handle \citet{AMS20}'s flagship example algorithmically, would require further support for dynamic sampling.
This is left for future work.
We can however express (and handle) concrete instances of this benchmark.
(We included \texttt{coupon-10}, \texttt{coupon-50} and \texttt{coupon-100}
as examples.)
Wrt.\ precision, the bounds generated by~\evimp\ are on the same order of magnitude as those generated by the tools from \ecoimp, while wrt.\ speed the generality of our tool takes its toll.

\paragraph*{Programs considered in this work.}

We have already discussed the examples \texttt{balls}, \texttt{throws}, \texttt{every-5} and
benchmark example \texttt{hire} in Section~\ref{Overview} and~\ref{Inference}, respectively.
Benchmark \texttt{every} constitutes the general case to procedure \pwhile!every!, where we do not restrict the number of bins to five, while \texttt{every-while} constitutes an encoding of this problem as loop program.
For the moment, \evimp cannot handle the general encoding of the question how many balls have to be thrown in average
until all bins are filled with at least one ball (compare~Section~\ref{Overview}) due to complexity of the example.
The remaining examples constitute recursive variants of standard examples and are given in full in the Supplementary Material.

\begin{experimenttable}{tab:evaluation:1}  
\benchmark{(a)~examples from~\citet{KatoenMMM10}}
biased-coin & $bool$ & $\sfrac{1}{2}$ & $0.029$ \\
binom-update & $x$ & $\sfrac{1}{2} \cdot \normf{N}$ & $0.012$ \\
uniform-dist & $g$ & \fail & $2.173$\\
uniform-dist-100 & $g$ & $99$ & $0.09$ 
\\[1ex]
\benchmark{(b)~examples from~\citet{ChakarovS14}}
mot-ex & $count$ & $\sfrac{44}{3}$ &  $0.02$
\\[1ex]
\benchmark{(c)~benchmark from~\citet{WangHR18}}
2d-walk & $count$ & $\normf{1+count}$ & $0.012$ \\
aggregate-rv & $x$ &
$\sfrac{1}{2} \cdot [x \geq -1 \land 499 \geq i] \cdot (1+x) + \sfrac{1}{2} \cdot [499 \geq i] \cdot \normf{x} + [i \geq 500] \cdot \normf{x}$
& $0.009$ \\
biased-coin & $x_1$ &
$\sfrac{1}{2} \cdot \normf{x_1} + \sfrac{1}{2} \cdot [x_1 > x_2]$
& $0.014$ \\
binom-update & ${x}$ &
$\sfrac{1}{4} \cdot [x \geq -1 \land 99 \geq n] \cdot (1+x) + \sfrac{3}{4} \cdot [99 \geq n] \cdot \normf{x} + [n \geq 100] \cdot x$
& $0.009$ \\
coupon5 & ${count}$ &
$[ count \geq -1 \land 4 \geq i] \cdot \normf{1+count} + [i \geq 5] \cdot \normf{count}$
& $0.009$ \\
eg-tail & ${x}$ &
$\normf{x} + \sfrac{19}{16} \cdot\normf{z} + \sfrac{7}{4}$ & $0.056$ \\
eg & ${x}$ & \fail & $0.18$ \\
hare-turtle & ${h}$ &
$ [h \geq -1 \land t \geq h] \cdot \sfrac{1}{22} \cdot \sum_{i=0}^{10} (h+i)  + [h > t] \cdot \normf{h} + [h \leq t] \sfrac{6}{11} \cdot \normf{h}$
& $0.010$ \\
hawk-dove & ${count}$ & $1$ & $0.013$ \\
mot-ex & ${count}$ & $\normf{1+count}$ & $0.009$ \\
recursive & ${x}$ & $\normf{x} + 10$ & $0.03$ \\
uniform-dist & ${g}$ &
$ \sfrac{1}{2} \cdot ([g \geq -\sfrac{1}{2} \land 9 \geq n] \cdot (1 + 2\cdot y) + [9 \geq n] \cdot 2 \cdot \normf{y}) + [n \geq 10] \cdot \normf{g}$
& $0.009$
\\[1ex]
\benchmark{(d)~benchmark from~\citet{BaoTPHR:2022}}
biasdir & ${x}$ & $\norm{1-x} + \normf{x}$ & $0.022$ \\
bin0 & ${x}$ & $\sfrac{1}{2} \cdot \normf{n} \cdot \normf{y} + \normf{x}$ & $0.075$ \\
bin1 & ${x}$ & $\sfrac{1}{2} \cdot \normf{M-n} + \normf{x}$ & $0.012$ \\
bin2 & ${x}$ & $\sfrac{1}{4} \cdot (\normf{n} + n^2) + \sfrac{1}{2} \cdot \normf{n} \cdot \normf{y} + \normf{x}$ & $0.141$ \\
deprv & ${z}$ & $\normf{z}$ & $0.011$ \\
detm & ${count}$ & $\normf{11-x} + \normf{count}$ & $0.012$ \\
duel & ${turn}$ & $\normf{turn} + 2 \cdot \normf{continuing}$ & $0.036$ \\
fair & ${count}$ & $\normf{count} + 2\cdot\normf{1-c_2}$ & $0.86$ \\
gambler0 & ${z}$ & $[x\geq 0 \land y \geq x] \cdot (x\cdot y -x^2) + \normf{z}$ & $0.063$ \\
geo0 & ${z}$ & $ \normf{1-flip} + \normf{z}$ & $0.014$ \\
\end{experimenttable}

\begin{experimenttable}{tab:evaluation:2}
\benchmark{(d)~benchmark from~\citet{BaoTPHR:2022} (cont'd)}
geo1 & ${z}$ & $\normf{1-flip} +  \normf{z}$ & $0.014$ \\
geo2 & ${z}$ & $\normf{1-flip} +  \normf{z}$ & $0.015$ \\
geoar0 & ${x}$ & $\sfrac{3}{2}(1 + \normf{1+y} + \normf{x})$ & $0.09$\\
linexp & ${z}$ & $\normf{z} + \sfrac{21}{8} \cdot \normf{n}$ & $0.014$ \\
mart & ${rounds}$ & $\normf{rounds} + 3 \cdot \normf{b}$ & $0.012$ \\
prinsys & ${bool}$ & $1$ & $0.013$ \\
revbin & ${z}$ & $\normf{z} + 2\cdot \normf{x}$ & $0.011$ \\
sum0 & ${x}$ & $\sfrac{1}{4} \cdot (\normf{n} + n^2)  + \normf{x}$ & $0.034$
\\[1ex]
\benchmark{(e)~examples from~\cite{KKMO:ACM:18}}
faulty\_sum & ${x}$ & $\sfrac{3}{7} \cdot \normf{x}^2 + \sfrac{4}{7} \cdot \normf{x} + 1$ & $0.04$ \\
geo & $\zerof$ & $0$ & $0.009$
\\[1ex]
\benchmark{(f)~benchmark from~\cite{AMS20}}
bridge & ${count}$ & $[b \geq x \land x \geq a] \cdot (-a \cdot b + a \cdot x + b \cdot x -x^2)$ & $0.067$ \\
coupon-10 & ${count}$ & $110$ & $0.145$ \\
coupon-50 & ${count}$ & $2550$ & $2.691$ \\
coupon-100 & ${count}$ & $10100$ & $10.63$ \\
nest-1 & ${count}$ & $4 \cdot \normf{n}$ & $0.014$ \\
nest-2 & ${count}$ & \fail & $100.00$ \\
trader-5 & $\normf{price}$ &
$\sfrac{20}{3} \cdot [price \geq -1] \cdot (1 + 2 \cdot price + price^2) + \sfrac{35}{6} \cdot \normf{1 + price}$
\\[1ex]
\benchmark{(g) examples from this work}
balls & ${b}$ & $\sfrac{1}{5} \cdot \normf{n}$ & $0.012$ \\
throws & ${throws}$ & $5$ & $0.012$ \\
every & ${number}$ & \fail & $100.0$ \\
every-5 & ${number}$ & $20$ & $0.017$ \\
every-while & ${number}$ & $25$ & $0.973$ \\
double\_recursive & ${y}$ & $0$ & $0.011$ \\
hire & ${hire}$ & $\normf{n}$ & $0.504$ \\
rdwalk & ${n}$ & $2\cdot \normf{n}$ & $0.016$ \\
rec1 & ${n}$ & $\sfrac{1}{2} \cdot (1+ \normf{n})$ & $0.014$ \\
\end{experimenttable}

\section{Related Work}
\label{Related}

Very briefly, we refer to the extensive
literature of analysis methods
for (non-determi\-nistic, imperative) probabilistic programs introduced in the last years.
These have been provided in the form of
\emph{abstract interpretations}~\citep{ChakarovS14};
\emph{martingales}, eg., ranking super-martingales~\citep{ACN:POPL:18,TOUH:ATVA:18,WFGCQS:PLDI:19}; or equivalently
\emph{Lyapunov ranking functions}~\citep{BG:RTA:05};
\emph{model checking}~\citep{Katoen16};
\emph{program logics}~\citep{MM05,KaminskiKatoen,KK:LICS:17,McIverMKK18,KKMO:ACM:18,NgoCH18,WangHR18,BaoTPHR:2022};
\emph{proof assistants}~\citep{BartheGB09};
\emph{recurrence relations}~\citep{Sedgewick96};
methods based on \emph{program analysis}~\citep{Kozen:JCSC:85,KatoenMMM10,CelikuM05}; or
\emph{symbolic inference}~\citep{GehrMV16};
and finally \emph{type systems}~\citep{ADG:LICS:19,WangKH20,LMZ:2022,VasilenkoVB22}.
%
In the following, we restrict our focus on related work concerned with the analysis of \emph{quantitative} of (non-determi\-nistic) probabilistic programs, notably to the areas of (automated) \emph{invariant} generation and \emph{expected cost analysis}.

\paragraph{Invariant generation.}

Generally, invariant generation is more challenging than the 
expected result value analysis that we study. Still, often \evimp\ derives exact bounds, thus establishing invariants. 
On the other hand, our methodology is applicable in a more general framework, for example to
expected cost analysis. Further, our methodology encompasses recursive (imperative) programs, which is---to the best of
our knowledge---not the case for any of the approaches on invariant generation
(or expected cost analysis, for that matter).
\citet{KatoenMMM10} provide constraint-based methods for the semi-automated generation of
linear quantitative invariants, based on sophisticated proof-based methods.
The studied examples are highly interesting and have been integrated into our benchmarks (see Section~\ref{Implementation}). The form of expectations considered can be very expressive.
We emphasise, however, that our method is fully automated, while the
approach in~\citep{KatoenMMM10} is only partly automated, in particular nested loops require user-interaction.
Related results have been reported by~\citet{ChakarovS14}, suitable adapting
an abstract interpretation framework to the notion of invariant generation. Their motivating
example can be handled by~\evimp\ fully automatically, establishing a slightly worse constant
bound than the exact bound (see Section~\ref{Implementation}).
In contrast to~\citep{KatoenMMM10,ChakarovS14}, \citet{WangHR18} and the very recent~\citet{BaoTPHR:2022}
provide fully automated (linear) invariant generation methodologies.
\citet{WangHR18} provide a compelling algebraic framework for the program analysis of probabilistic programs.
Apart from the here relevant linear invariant generation,
interprocedural Bayesian inference analysis and the Markov decision problem are conducted. 
Wrt.\ linear invariant generation, we have considered
all provided examples in our benchmarks and obtained comparable results, while improving the speed
of the analysis by a magnitude (in comparison to the analysis times reported in~\cite{WangHR18}), cf.~Section~\ref{Implementation}.
Automation of the method developed in~\cite{WangHR18} is based, like ours, on a template approach. To overcome
the dependency on templates, \citet{BaoTPHR:2022} have developed a highly interesting learning approach.
Their approach is data-driven.
Apart from invariants, Bao et al.\ also consider \emph{sub-invariants} which are dual to our upper invariants, establishing
lower bounds on the pre-expectations. In both cases, however, analysis times are high.
We have incorporated the benchmark examples from~\cite{BaoTPHR:2022} into our
test suites. In all cases our tool~\evimp\ provides precise invariants, cf.~Section~\ref{Implementation}.
This is remarkable, as reported in~\cite{BaoTPHR:2022}, their prototype implementation \tool{Exist}
handles only 12/18 of their benchmark suite fully automatically.
Further, as mentioned, our analysis times are in the milliseconds range for all benchmarks.

\paragraph{Expected cost analysis.}

\citet{KKMO:ACM:18} establishes an expected cost analysis of recursive programs and we have suited
the corresponding two example to our benchmark suite. Both examples can be handled fully automatically. Technically
our development of recursive programs constitutes an extension as our language (and methodology) admits local variables,
formal parameters and (unrestricted) return values. This allows a more natural representation of programs (see Section~\ref{PWhile}). Still the definition of our expectation transformer $\et{\fn}$ is closely related to the corresponding definition in~\cite[Chapter~7]{KKMO:ACM:18}.
We emphasise, however, that in~\cite{KKMO:ACM:18} automation is discussed only superficially.
\citet{AMS20}, on the other hand, take automation very seriously. As mentioned above, we have
taken inspiration from their work in the modular analysis of recursion-free programs. Despite our efforts, however,
our prototype implementation~\evimp\ lacks scalability and speed in comparison to their tool~\ecoimp. Wrt.\ precision, however,
we often derive the same bounds on expected costs. Further, and as argued, our methodology focusing
on expected value analysis is more general and our implementation incorporated the highly non-trivial
handling of recursive programs.
Technically the closest comparison is to work on amortised cost analysis of functional languages (eg.~\cite{WangKH20,LMZ:2022}), as resource parametricity~\cite{Hoffmann11} has been studied in this context. In comparison to \citet{WangKH20}, our
tool~\evimp\ cannot infer higher-order moments but concerning expectations seems to have better support for recursion. For instance, we were not able to reproduce an expected cost analysis of the \pwhile!balls! procedure in their~\raml\
prototype.
Similarly, \evimp\ is no match to \citet{LMZ:2022}'s \atlas, when it comes to the precise analysis of (probabilistic) data
structures. Remarkably, however, the \pwhile!balls! benchmark can only be expressed
convolutedly in their language. (Due to the lack of support for general, inductive, data structures.) Further, their automated analysis does not derive an optimal bound.

\paragraph{Expected value analysis.}

Finally, we briefly remark on very recent and partly motivating work by~\citet{VasilenkoVB22}.
In~\citep{VasilenkoVB22} a refinement type system---Liquid Haskell, cf.~\cite{Vazou16,HVH:2020}---is updated, to reason about relational properties of probabilistic computations.
One of the (simple) examples studied is equivalent to procedure \pwhile!balls!, depicted in Listing~\ref{fig:motivation}(\subref{lst:balls}) and~\citet{VasilenkoVB22} provide a semi-automated proof that the expected value of $b$ is $p \cdot n$ is provided. No attempt at full automation is made.

\section{Conclusion}
\label{Conclusion}

We established a fully automated expected result value analysis for probabilistic programs in the presence of natural programming constructs, in particular recursion. Our analysis is in the form of an \emph{expectation transformer}
$\et{\cdot}$ and its syntactic representation $\ET{\cdot}$. As argued, automated inference of upper invariants
is challenging for \PWHILE, due to the presence of recursion. We have overcome these challenges and implemented
the established methodology in our novel prototype implementation~\evimp.

In future work, we aim to incorporate
\begin{inparaenum}[(i)]
\item \emph{more program features} and , like eg.\ support for \emph{dynamic} uniform distributions;
\item improve the \emph{constraint solving} capabilities of our prototype implementation, to handle the
  analysis of further natural probabilistic programs and data structures fully automatically.
\end{inparaenum}
 
\begin{acks}
  We would like the thank the annoymous reviewers for their work and invaluable suggestions, which greatly
  improved our presentation. This work is partly supported by the \grantsponsor{}{INRIA Associate Team}{} \grantnum{}{TCPRO3}.  
\end{acks}

\section{Availability of Data and Software}

Our prototype implementation~\evimp is publicly available via the following Zenodo link: \url{https://doi.org/10.5281/zenodo.7706691}. The artifact~\cite{AMS23b} is given as a Docker image running a minimal Linux distribution, containing the sources in directory \text{/evimp}, as well as pre-compiled executables. A second, larger, image contains the source distribution as well as the complete toolchain to compile~\evimp. We have successfully tested the images under Linux, MacOS and Windows, as long as these were non-ARM architectures. For ARM64 architectures, the artifact additionaly provides instructions to install from source. Detailed instructions on the use of~\evimp are provided as well.

\bibliographystyle{ACM-Reference-Format}

\newpage
\appendix

\section{Mathematical Background}

\begin{proposition}[Function Lifting of $\omega$-CPOs {\citep[Section 8.3.3]{Winskel:93}}]\label{p:fun-omega}
  Let $(D,\sqsubseteq)$ be an $\omega$-CPO.
  Then $(A \to D, \bm{\sqsubseteq})$ where $\bm{\sqsubseteq}$ extends $\sqsubseteq$
  point-wise forms an $\omega$-CPO, with the supremum on $A \to D$ given point-wise.
  If $\sqsubseteq$ has least and greatest elements $\bot$ and $\top$, then $\bm{\bot}$ and $\bm{\top}$
  are the least and greatest elements of $\bm{\sqsubseteq}$, respectively.
\end{proposition}


\begin{theorem}[Kleene's Fixed-Point Theorem for $\omega$-CPOs, {\citep[Theorem 5.11]{Winskel:93}}]\label{t:fixed-point}
  Let $(D,\sqsubseteq)$ be a $\omega$-CPO with least element $\bot$. Let $\chi \colon D \to D$
  be continuous (thus monotone).
  Then $\chi$ has a least fixed-point given by
  \[
    \lfp \chi = \sup_{n \in \Nat} \chi^n(\bot) \tpkt
  \]
\end{theorem}

\begin{lemma}[Continuity of Expectation]\label{l:continuity-E}
  $\E{\mu}{\sup_{n \in \Nat} f_n} = \sup_{n \in \Nat} \E{\mu}{f_n}$
  for all $\omega$-chains $(f_n)_{n \in \Nat}$.
\end{lemma}
\begin{proof}
  This is the discrete version of Lebesgue's Monotone Convergence Theorem~\citep[Theorem 21.38]{Schlechter:96}.
\end{proof}

\section{Omitted Proofs}

\again{p:laws}
\begin{proof}
  The proofs follow the pattern of the proof of continuity of the expected cost transformer in~\cite[Lemma~6.2]{AMS20}.  
\end{proof}

We define finite approximations of procedure environments $\et{\prog}^{(i)}$ inductively,
so that $\et{\prog}^{(0)} \defsym \lam{\fn}{\lam{f\,\vec{\val}\,\mem}{0}}$
and $\et{\prog}^{(i+1)} \defsym \lam{\fn}{\et[\et{\prog}^{(i)}]{\fn}}$, where $\fn \in \prog$.
In this way $\et{\prog}^{(0)}$ represents the poorest approximation by the constant zero function, while approximations
are refined iteratively.

\begin{proposition*}[Finite Approximations of Procedure Environments]{p:approximation}
  For any program $\prog$, we have $\et{\prog} = \sup_{i \geqslant 0} \et{\prog}^{(i)}$.
\end{proposition*}
\begin{proof}
  Direct consequence of the continuity of $\et{\fn}$ and Knaster-Tarski fixed-point theorem.
\end{proof}

\again{t:soundness}
\begin{proof}
  Let $\termtwo, \term \in \Term{V}{Z}$ and $\theta \colon Z \to \Val$. Then we prove for all commands
  $\cmd$ that
  \begin{equation}
    \label{eq:aux}
     \et[\aenv][\sem{\termtwo \theta}]{\cmd} \app \sem{\term \theta} \leq \sem{(\ET[\penv][\termtwo]{\cmd} \app \term) \theta}
    \tpkt
  \end{equation}
  Let $s \defsym \sem{\termtwo \theta}$ and
  $t \defsym \sem{\term \theta}$. In order to prove~\eqref{eq:aux}, we prove the following,
  for all $i \in \Nat$:
  \begin{equation*}
    \et[\aenv_i][s]{\cmd} \app t \leq \sem{(\ET[\penv][\termtwo]{\cmd} \app \term) \app \theta} 
    \tkom
  \end{equation*}
  where
  \begin{inparaenum}[(i)]
  \item $\aenv_0 \app \fn \defsym \lambda \fn. 0$;
  \item $\aenv_{i+1} \app \fn \defsym \et[\aenv_i]{\fn}$;
  \end{inparaenum}
  and $\aenv \defsym \sup_{i \geqslant 0} \aenv_i$.

  In proof, we elide the logical context $\ctx[\fn]$, employed in the
  definition of $\ET{\fn}$ for notational convenience. As the essence of
  this context information is that logical variables are always
  instantiated non-negatively, this can be guaranteed globally.
  
  We proceed by main induction on $i$ and side induction on $\cmd$, where we focus on the (main) step case, as the case for $i=0$ is similar, but simpler. Thus let $i > 0$ and we
  proceed by case induction on $\cmd$, considering the most interesting cases, only.
  \begin{proofcases}
    \proofcase{$\cskip$}
    By unfolding of definitions, we easily obtain 
    \begin{equation*}
      \et[\aenv_{i+1}][s]{\cskip} \app t = \sem{\term \app \theta} =
      \sem{(\ET[\aenv_{i+1}][\termtwo]{\cskip} \app \term) \app \theta}
      \tkom
    \end{equation*}
    which concludes the case.
        
    \proofcase{$\vx \sample \ccall{\fn}{\vec{\expr}}$}
    Suppose $\fn \in \Prog$. By unfolding of definitions, we obtain for $\mem \in \Mem{V}$
    \begin{align*}
      \et[\aenv_{i+1}][s]{\vx \sample \ccall{\fn}{\vec{\expr}}} \app t \app \mem
      & = \aenv_{i+1} \app \fn \app \underbrace{(\lam{\val\,\memtwo}{t \app (\meml{\mem} \uplus \memg{\memtwo}) [\vx \mapsto \val]})}_{\defsym f} \app (\sem{\vec{\expr}} \app \mem) \app \memg{\mem}
      \\
      & = \et[\aenv_i]{\fn} \app f \app (\sem{\vec{\expr}} \app \mem) \app \memg{\mem}
      \\
      & = \et[\aenv_i][f]{\bdy{\fn}} \app (\lam{\memtwo}{t \app (\meml{\mem} \uplus \memg{\memtwo}) [\vx \mapsto 0]} \app
        \memg{\mem} \uplus \{\params{\fn} \mapsto \sem{\vec{\expr}} \app \mem \})
      \\
      & \leq \et[\aenv_i][\sem{\termk{\fn} \app \theta}]{\bdy{\fn}} \app \sem{\termk{\fn}[\lvr \mapsto 0] \app \theta} \app
        \memg{\mem} \uplus \{\params{\fn} \mapsto \sem{\vec{\expr}} \app \mem \}
      \\
      & \leq \sem{(\ET[\aenv_{i+1}][\termk{\fn}]{\bdy{\fn}} \app \termk{\fn}[\lvr \mapsto 0]) \app \theta} \app
       \memg{\mem} \uplus \{\params{\fn} \mapsto \sem{\vec{\expr}} \app \mem \}
      \\
      & = \sem{(\ET[\aenv_{i+1}][\termk{\fn}]{\bdy{\fn}} \app \termk{\fn}[\lvr \mapsto 0]) [\params{\fn} \mapsto \lvas{\fn}] \app \theta} \app
        \memg{\mem} \uplus \{\lvas{\fn} \mapsto \sem{\vec{\expr}} \app \mem \}
      \\
      & = \sem{(\ET{\fn} \app \termk{\fn}) \app \theta} \app
        \memg{\mem} \uplus \{\lvas{\fn} \mapsto \sem{\vec{\expr}} \app \mem \}
      \\
      & \leq \sem{\termh{\fn} \app \theta} \app \memg{\mem} \uplus \{\vec{\lva}_{\fn} \mapsto  \sem{\vec{\expr}} \app \mem \}
      \\
      & = \sem{\termh{\fn} [\vec{\lva}_{\fn} \mapsto \vec{\expr}] \app \theta} \app \mem
        = \sem{(\ET[\aenv_{i+1}][\termtwo]{\vx \sample \ccall{\fn}{\vec{\expr}}} \app \term \app \theta} \app \mem
        \tpkt
    \end{align*}
    Here, we have used in conjunction with reduction to definitions
    \begin{inparaenum}[(i)]
    \item two instances of the assumed side-condition $\vdash \term[\vx \mapsto \lvr] \leq \termk{\fntwo} \app \theta$ in line three;
    \item together with monotonicity of the expectation transformer $\et{\cmd}$, cf.~Figure~\ref{fig:et-idents};
    \item induction hypothesis in line four; and finally
    \item in the pre-ultimate line, the main constraint on the soundness of the inference mechanisms~\eqref{eq:main-constraints}.
    \end{inparaenum}
    This concludes the case.

    \proofcase{$\cret{\expr}$}
    By unfolding of definitions, we obtain for $\mem \in \Mem{V}$
    \begin{align*}
      \et[\aenv_{i+1}][s]{\cret{\expr}} \app t \app \mem & = \sem{\termtwo \app \theta} \app (\sem{E} \app \mem) \app \memg{\mem}
      \\
                                                         & = \sem{\termtwo [\lvr \mapsto \expr] \app \theta} \app \mem
      \\
                                                         & = \sem{(\ET[\aenv_{i+1}][\termtwo]{\cret{\expr}} \app \term) \app \theta} \app \mem
                                                           \tpkt
    \end{align*}
    This concludes the case.

    \proofcase{$\clet{\vx}{\expr}{\cmd}$}
    By unfolding of definitions in conjunction with  application of the induction hypothesis, we obtain for $\mem \in \Mem{V}$
    \begin{align*}
      \et[\aenv_{i+1}][s]{\clet{\vx}{\expr}{\cmd}} \app t \app \mem & = \et[\aenv_{i+1}][s]{\cmd} \app t \app \mem [\vx \mapsto \sem{\expr}\app\mem]
      \\
                                                                    & \leq \sem{(\ET[\aenv_{i+1}][\termtwo]{\cmd} \app \term) \app \theta} \app \mem [\vx \mapsto \sem{\expr}\app\mem]
      \\
                                                                    & = \sem{(\ET[\aenv_{i+1}][\termtwo]{\cmd} \app \term) [\vx \mapsto \expr] \app \theta} \app \mem
      \\
      & = \sem{(\ET[\aenv_{i+1}][\termtwo]{\clet{\vx}{\expr}{\cmd}} \app \term) \app \theta} \app \mem
    \end{align*}
    In the third line, we employ that due to the variable condition the local variable $\vx$ is distinct from all
    (global) variables in the domain of $\mem$. This concludes the case.

    \proofcase{$\cmd \sep \cmdtwo$}
    By unfolding of definitions in conjunction with two applications of the induction hypothesis, we obtain for $\mem \in \Mem{V}$
    \begin{align*}
      \et[\aenv_{i+1}][s]{\cmd \sep \cmdtwo} \app t \app \mem & = \et[\aenv_{i+1}][s]{\cmd} \app (\et[\aenv_{i+1}][s]{\cmdtwo} \app t) \app \mem
      \\
                                                              & \leq \sem{(\ET[\aenv_{i+1}][\termtwo]{\cmd} \app \term) \app \theta} \app
                                                                \sem{(\ET[\aenv_{i+1}][\termtwo]{\cmdtwo} \app \term) \app \theta} \app \mem 
      \\
                                                              & = \sem{\ET[\aenv_{i+1}][\termtwo]{\cmd} \app ( \ET[\aenv_{i+1}][\termtwo]{\cmdtwo} \app \term) \app \theta} \app \mem
      \\
      & = \sem{(\ET[\aenv_{i+1}][\termtwo]{\cmd \sep \cmdtwo} \app \term) \app \theta} \app \mem
    \end{align*}
    Apart from applications of the induction hypothesis, in line two, we have employed monotonicity,
    cf.~Figure~\ref{fig:et-idents}. This concludes the case.

    \proofcase{$\cif{\bexpr}{\cmd}{\cmdtwo}$}
    By unfolding of definitions in conjunction with two applications of the induction hypothesis, we obtain for $\mem \in \Mem{V}$
    \begin{align*}
      \et[\aenv_{i+1}][s]{\cif{\bexpr}{\cmd}{\cmdtwo}} \app t \app \mem & = [\sem{\bexpr} \app \mem] \cdot \et[\aenv_{i+1}][s]{\cmd} \app t \app \mem + [\sem{\neg\bexpr} \app \mem] \cdot \et[\aenv_{i+1}][s]{\cmdtwo}\app t \app \mem
      \\
                                                                        & \leq [\sem{\bexpr} \app \mem] \cdot \sem{(\ET[\aenv_{i+1}][\termtwo]{\cmd} \app \term) \app \theta} + [\sem{\neg\bexpr} \app \mem] \cdot \sem{(\ET[\aenv_{i+1}][\termtwo]{\cmdtwo} \app \term) \app \theta}
      \\
                                                                        & = \sem{([\bexpr] \cdot \ET[\aenv_{i+1}][\termtwo]{\cmd} \app \term) + [\neg{\bexpr}] \cdot \ET[\aenv_{i+1}][\termtwo]{\cmdtwo} \app \term ) \app \theta} \app \mem
      \\
      & = \sem{(\ET[\aenv_{i+1}][\termtwo]{\cif{\bexpr}{\cmd}{\cmdtwo}} \app \term) \app \theta} \app \mem
    \end{align*}
    This concludes the case.
    
    \proofcase{$\cwhile{\bexpr}{\cmdtwo}$}
    In this case,
    $(\ET[\penv][\termtwo]{\cmd} \app \term)\app \theta = \termthree \app \theta$
    for some term $\termthree$, satisfying
    \begin{inparaenum}[(i)]
      \item $\bexpr \vdash \ET[\penv][\termtwo]{\cmd} \app \term \leq \termthree$ and
      \item $\neg\bexpr \vdash \term \leq \termthree$.
    \end{inparaenum}
    By induction hypothesis and monotonicity of $\et[\penv][\termtwo]{\cmdtwo}$, this says nothing more than
    that $\sem{\termthree \app \theta}$ is a loop-invariant for the while loop, wrt.\ $\sem{\term \app \theta}$ (see \Cref{fig:et-idents}).

    \proofcase{$\cnd{\cmd}{\cmdtwo}$}
    In this case,
    $(\ET[\penv][\termtwo]{\cmd} \app \term)\app \theta = \termthree \app \theta$
    for some term $\termthree$, satisfying
    \begin{inparaenum}[(i)]
    \item $\vdash \ET[\penv][\termtwo]{\cmd} \app \norm \leq \termthree$ and
    \item $\vdash \ET[\penv][\termtwo]{\cmdtwo} \app \norm \leq \termthree$.
    \end{inparaenum}
    The case follows by of two applications of induction hypothesis.
  \end{proofcases}

\end{proof}

\section*{Additional Benchmarks}

We detail in Figures~\ref{fig:additional} and~\ref{fig:additional2} the benchmarks from this paper, whose evaluation results are given in Table~\ref{tab:evaluation:2}~(h).

\begin{figure*}
\caption{Additional Benchmark Examples}
\label{fig:additional}
\smallskip
\centering
\begin{subfigure}{0.45\linewidth}
\begin{lstlisting}[style=pwhile,style=framed]
def f(x):
  var y
  if (x $\geq$ 0) {
    if (Bernoulli(\frac{1}{2})) {
      y $\sample$ f(x) }
    if (Bernoulli(\frac{1}{3})) {
      y $\sample$ f(y) }
    };
  return y
\end{lstlisting}
\caption{\texttt{double\_recursive} example.}
\end{subfigure}
\quad
\begin{subfigure}{0.45\linewidth}
\begin{lstlisting}[style=pwhile,style=framed]
def rdwalk(n):
  var c $\ass$ 0  
  if (n > 1) {
    if (Bernoulli(\frac{1}{2})) {
      c $\sample$ rdwalk(n-2) }
    else {
      c $\sample$ rdwalk(n+1) }
    return c+1 }
  else { return c }
\end{lstlisting}
\caption{\texttt{rdwalk} benchmark example.} 
\end{subfigure}
\end{figure*}

\begin{figure*}
\caption{Additional benchmark examples (cont'd)}
\label{fig:additional2}
\smallskip
\centering
\begin{subfigure}[b]{0.45\linewidth}
\begin{lstlisting}[style=pwhile,style=framed]
def f(n):
 var m
 if (n > 0) {
   m $\sample$ f(n - 1)
 };
 m $\sample$ m + Bernoulli($\frac{1}{2}$);
 return m
 
\end{lstlisting}
\caption{\texttt{rec1} benchmark example.}
\end{subfigure}
\quad
\begin{subfigure}[b]{0.45\linewidth}
\begin{lstlisting}[style=pwhile,style=framed]
def every(bins):
   var d, number, k
   k $\ass$ 1   
   while (k $\leq$ 5) {     
      if (Bernoulli($\frac{5-k+1}{5}$) {k $\ass$ k + 1}
      number $\ass$ number + 1
   };
   return number 
\end{lstlisting}
\caption{\texttt{every-while} benchmark example.}
\end{subfigure}
\end{figure*}

\end{document}